\pdfoutput=1
\documentclass[twoside,11pt]{article}
\usepackage{amsmath,mathtools}
\usepackage{enumitem}
\usepackage{array}
\usepackage{upgreek}
\usepackage{jmlr2e}
\renewenvironment{proof}[1][Proof]{\par\noindent{\bf #1\ }}{\hfill\BlackBox\\[2mm]}
\hypersetup{
  pdftitle={Serial Depth of Adaptive Parallel Sampling: Lower Bounds and
            Schedules for Masked Diffusion Models},
  pdfauthor={Chuling Wen, Weijie Liang, Jian Lu, Zhouchen Lin},
  pdfsubject={Information theory; parallel sampling; masked diffusion models},
  pdfkeywords={masked diffusion language models, parallel decoding,
               conditional total correlation, round complexity,
               serial depth}}

\newtheorem{openproblem}[theorem]{Open Problem}

\newcommand{\KL}[2]{D_{\mathrm{KL}}\!\left(#1 \,\middle\|\, #2\right)}
\newcommand{\TC}{\mathrm{TC}}
\newcommand{\tc}{\mathrm{tc}}
\newcommand{\cE}{\mathcal{E}}
\newcommand{\E}{\mathbb{E}}
\newcommand{\Bin}{\mathrm{Bin}}
\newcommand{\eps}{\varepsilon}

\ShortHeadings{Serial Depth of Adaptive Parallel Sampling}{Wen, Liang, Lu and Lin}
\firstpageno{1}

\editor{}

\begin{document}

\title{Serial Depth of Adaptive Parallel Sampling:\\
Lower Bounds and Schedules for Masked Diffusion Models}

\author{\name Chuling Wen \email chulwen@163.com \\
        \addr Shenzhen Key Laboratory of Advanced Machine Learning and
        Applications\\
        School of Mathematical Sciences, Shenzhen University\\
        Shenzhen 518060, China
        \AND
        \name Weijie Liang \email liangweijie2022@email.szu.edu.cn \\
        \addr Shenzhen Key Laboratory of Advanced Machine Learning and
        Applications\\
        School of Mathematical Sciences, Shenzhen University\\
        Shenzhen 518060, China
        \AND
        \name Jian Lu\thanks{Corresponding author.} \email jianlu@szu.edu.cn \\
        \addr Shenzhen Key Laboratory of Advanced Machine Learning and
        Applications\\
        School of Mathematical Sciences, Shenzhen University\\
        Shenzhen 518060, China\\
        National Center for Applied Mathematics Shenzhen (NCAMS)\\
        Shenzhen 518055, China
        \AND
        \name Zhouchen Lin \email zlin@pku.edu.cn \\
        \addr School of Intelligence Science and Technology\\
        Peking University\\
        Beijing 100871, China}

\maketitle

\begin{abstract}
Masked diffusion models generate sequences in rounds, predicting several
masked positions at once. Fewer rounds can accelerate decoding, but sampling
positions independently can lose their dependence and bias the output. How
few rounds suffice for a prescribed error budget? Beyond error bounds for
individual decoding rules, we study sharp lower bounds on the number of
rounds required for a given target distribution. We introduce
\emph{serial depth}: the minimum number of rounds, averaged over target
sequences, needed to meet a forward Kullback--Leibler error budget. With
exact conditional probabilities, we prove these bounds for all deterministic
rules that may adapt to generated values and construct matching schedules
for several distribution families, with depths ranging from $\log\log n$ to
linear in sequence length $n$. To narrow the gap to trained models, we show
that these lower bounds remain valid for learned predictors and give
sufficient accuracy conditions on visited contexts for schedules to retain
their error guarantees. On two pretrained models across five text domains, a
schedule suggested by the theory achieves higher average log-likelihood on
held-out text in five rounds than random selection in eight. A preregistered
confirmation study certifies a relative improvement on independent dialogue
data.
\end{abstract}


\section{Introduction}

Masked diffusion language models generate a sequence in rounds
\citep{austin2021,radd2024,llada2025}. Each round is one model call that
fills several masked positions. Each position is sampled from the model's
conditional distribution given the positions filled so far. Filling more positions per round
reduces the number of serial model calls, but generation quality can
deteriorate as more positions are filled together
\citep{parallelbench2026,whystruggle2026}. One reason is that
positions filled in the same round are sampled independently. Two
positions that must agree can then disagree even if each marginal is
exact. Filling them in different rounds avoids this error but costs a
round. How few rounds suffice for a given target distribution and error
budget, and which positions should share a round?

Existing analyses evaluate particular rules or restricted classes of
rules. \citet{entropyunmask2025} decompose the error of an adaptive rule
into prediction error and the dependence within its batches, with equality
for deterministic rules. Later work bounds the error and the expected
round count of a random-order scan stopped by an entropy budget
\citep{cai2026confidence}. It also lower-bounds the rounds of every rule
that commits only positions meeting a batch-size-dependent confidence
threshold \citep{fu2025bits}. Other work optimizes round
sizes among value-independent schedules
\citep{lavenant2025,chen2025schedules}.
Evaluating one rule or optimizing within a
restricted class does not show whether another rule could use fewer rounds. Rules that adapt to the generated values form a much larger class. For a given target, the minimum number of rounds over this class has not
been determined. Schedules that first fill separated positions
\citep{dus2025,bansal2026} and the error of fixed partitions
\citep{parallelbench2026} have been studied. What is new here is the
optimal spacing of such schedules, lower bounds over all adaptive rules,
and the asymptotic error--round frontier.

We introduce \emph{serial depth} to quantify this minimum. For a target law
$p$ on length-$n$ sequences, a deterministic value-adaptive rule $\pi$
chooses the next positions from the values already filled and samples them
independently from their exact conditional distributions. If its output law
is $q_\pi$ and it uses $R(x)$ rounds on $x$, set
\[
 D_\eps(p):=\min_{\pi:\,\KL{p}{q_\pi}\le\eps}\E_{x\sim p}R(x).
\]
Here $\E_{x\sim p}$ denotes expectation over target sequences $x$ drawn
from $p$, so rounds are averaged over target sequences. The budget $\eps$
is the permitted total forward Kullback--Leibler (KL) error in bits. Our starting point is the decomposition of \citet{entropyunmask2025},
which for a deterministic rule with exact conditionals reads
\[
 \KL{p}{q_\pi}
 =\E_{x\sim p}\!\left[\sum_{t=1}^{R(x)}
       \TC(X_{S_t}\mid x_{C_t})\right],
\]
where $S_t$ is the batch filled in round $t$, $x_{C_t}$ is the context
filled before it, and the \emph{conditional total correlation}
\[
 \TC(X_S\mid x_C)=\sum_{i\in S}H(X_i\mid x_C)-H(X_S\mid x_C)
\]
is zero exactly when the positions in $S$ are conditionally independent
given $x_C$. It thus measures the dependence lost when a batch is sampled
independently. \citet{entropyunmask2025} use this decomposition, with an entropy upper
bound on its dependence term, to design one rule. The challenge here is to derive round lower
bounds that remain valid when both the selected positions and the stopping
time depend on the observed values; we establish such bounds for
structured target distributions and construct matching schedules. Section~\ref{sec:prelim} specifies the model, and
Remark~\ref{rem:randomized} explains why randomized rules pose a different
problem.

Our contributions are as follows.
\begin{itemize}
\item We introduce serial depth, the fewest rounds that any deterministic
adaptive rule needs for a given target and error budget, and characterize
it for families ranging from local to global dependence. We prove lower bounds that hold for every rule and construct schedules
that match them. Depth ranges from $\log\log n$ to linear in $n$. It is
governed by the dependence structure rather than by entropy or likelihood
(Sections~\ref{sec:main} and~\ref{sec:families}).
\item We derive an optimal schedule for binary mixing chains, in which
dependence decays with distance. At a fixed positive error budget, filling widely spaced anchors first
needs only about $\log\log n$ rounds and is optimal up to an additive
constant (Theorem~\ref{thm:mixing}). On the same chains,
confidence-ordered selection with leftmost tie-breaking needs a number of
rounds linear in $n$ (Theorem~\ref{thm:confchain}).
\item We connect the theory to trained models. The depth lower bounds remain
valid for learned predictors, and sufficient accuracy conditions on the
visited contexts let a schedule keep its error guarantee
(Section~\ref{sec:nonoracle}). On two pretrained masked diffusion models
across five text domains, the anchor schedule achieves higher average log-likelihood in five rounds
than random selection in eight. A preregistered study on independent
dialogue data certifies this relative improvement. The five-round decoder
also runs about $1.5\times$ faster
(Section~\ref{sec:decoders-part}).
\end{itemize}

\begin{table}[t]
\centering
\footnotesize
\setlength{\tabcolsep}{3pt}
\begin{tabular}{>{\raggedright\arraybackslash}p{0.17\textwidth}>{\raggedright\arraybackslash}p{0.28\textwidth}>{\raggedright\arraybackslash}p{0.35\textwidth}>{\raggedright\arraybackslash}p{0.14\textwidth}}
\hline
target family & proved depth bounds & schedule attaining the upper bound & source\\
\hline
mixing binary Markov chain & $D_\eps=\log_2\log_2n+O(1)$ for fixed $\eps>0$;
$D_0=\lceil\log_2(n+1)\rceil$ & for $\eps>0$, reveal every $\ell$th position with $\ell\asymp\log n$, then bisect all gaps in parallel; for $\eps=0$, bisect from the midpoint & Theorem~\ref{thm:mixing}\\
order-$m$ Markov & $D_0\le m(\lceil\log_2(n/m)\rceil+1)$ & reveal a middle block of $m$ positions one position per round, then recurse on both sides in parallel & Proposition~\ref{thm:markov}\\
Bernoulli walk & $D_\eps=\Theta(\log n)$ for every fixed $\eps\ge0$ & reveal the middle position, then recurse on both halves in parallel; left-to-right blocks need $\Theta(n)$ rounds & Proposition~\ref{thm:ordergap}, Theorem~\ref{thm:walkconverse}\\
balanced binary strings & $D_\eps=\Theta(\log^2n)$ for fixed $\eps>0$ & reveal a fraction $\gamma\asymp\eps/\log n$ of the unrevealed positions per round, then the last $O(\log n)$ one at a time & Theorem~\ref{thm:balanced}\\
binary one-hot blocks & $D_\eps=\Theta(\sqrt n)$ for every fixed $\eps\ge0$ & reveal one position of every block per round & Theorem~\ref{thm:onehot}\\
binary codebook mixture & $D_\eps=\Theta(n/\log n)$ for every fixed $\eps>0$ & value-adaptive codebook schedule (information-theoretic) & Theorem~\ref{thm:binarynear}\\
uniform random permutations & $D_\eps=(r_*(\eps)+o(1))\,n$ for every fixed $\eps\ge0$, for example $0.75\,n$ at $\eps=\tfrac12$ & batch sizes in nonincreasing order, given by the dynamic program of Proposition~\ref{prop:dporder} & Theorems~\ref{thm:perm} and~\ref{thm:frontier}\\
\hline
\end{tabular}
\caption{Depths and attaining schedules. Lower bounds hold for all deterministic
value-adaptive rules; a lone $\le$ reports only an upper bound. All listed
schedules are fixed except the codebook construction, whose matching upper
bound concerns target-averaged depth. Bounds hold on the lengths defining each
family: even $n$ for balanced strings, $n=m^2$ for one-hot blocks, and powers of
two $n\ge256$ for codebooks. Constants may depend on the fixed error budget and
family parameters. For permutations,
$r_*(\eps)=(1+2^{-2\eps})/2$ when
$0\le\eps\le\tfrac12\log_23$ (Theorem~\ref{thm:frontier}).}

\label{tab:depth-regimes}
\end{table}

Section~\ref{sec:prelim} defines the reveal model and serial depth.
Section~\ref{sec:main} states the cost identity and shows how it yields
depth bounds, and Section~\ref{sec:families} derives the depth of
structured families. Section~\ref{sec:nonoracle} treats learned
conditionals and prediction error,
Section~\ref{sec:decoders-part} reports experiments, and
Section~\ref{sec:discussion} discusses scope and open problems. Proofs are
in the appendices.

\subsection{Related Work}\label{sec:related}

Prior work concerns practical decoding rules, error decompositions,
analyses of particular schedules, and sampling of unknown targets; we
discuss them in this order.
Mask-Predict uses confidence to choose tokens for revision
\citep{maskpredict2019}. Masked (absorbing-state) diffusion models predict every masked position in
each call \citep{sedd2024,radd2024,zhaoding2025,llada2025}, so a parallel
decoder must choose which predictions to commit. LLaDA selects confident
positions; Fast-dLLM combines thresholds with representation reuse
\citep{llada2025,fastdllm2025}. Entropy budgets
\citep{entropyunmask2025}, conditional-independence tests
\citep{azangulov2025punt}, and attention-based corrections
\citep{sahin2026adas} address dependence that individual confidence misses.
Our depth bounds ask how well any deterministic commitment rule can do.

Value-independent schedules ignore the generated values but may be random.
For such schedules, \citet[Lemmas~D.3 and~D.8]{fenggeng2025} and
\citet[Eq.~(23)]{licai2025} write or bound parallelization error through
the dependence within each batch. \citet{entropyunmask2025} give such a
decomposition for adaptive rules that also includes model error. Their
decomposition holds with equality for deterministic rules and is the
identity used in Section~\ref{sec:main}. \citet{turok2026duel} compute exact likelihoods for
deterministic unmasking decoders. \citet{xu2026scheduling} compare parallel
and sequential prediction for fixed blocks; Corollary~\ref{cor:seq}
extends that comparison to adaptive reveal sets under the target law.

Schedule analyses optimize round sizes under uniform position selection
\citep{lavenant2025,chen2025schedules}, stop random-order scans according
to entropy budgets \citep{cai2026confidence}, or certify the error from
samples \citep{wainwright2026}. Related analyses of discretized samplers
bound the number of steps through total correlation, dual total
correlation, or an effective total correlation
\citep{zhao2026,dmitriev2026,dmitriev2026adaptive}. Dilated unmasking
(DUS) fixes in advance a coarse-to-fine schedule within each block of
positions and motivates it by entropy bounds for an exact denoiser on
mixing Markov sources \citep{dus2025}. For mixing chains,
Theorem~\ref{thm:mixing} shows that anchors at spacing $\Theta(\log n)$
followed by bisection need only $\log_2\log_2n+O(1)$ rounds, and that no
rule needs fewer. Generation order has also been analyzed theoretically
and measured empirically \citep{zhang2026,zhong2026order,bansal2026}. These results
analyze specified mechanisms. In particular, a randomized reveal order
is a mixture of deterministic rules, so deterministic and randomized
round bounds do not automatically transfer to one another
(Remark~\ref{rem:randomized}).

\citet[Theorems~1--3]{parallelbench2026} bound the error of decoding along
a partition fixed before generation by a sum of conditional total
correlations and show that the minimum of this bound over partitions does
not increase with the number of steps. Their shuffle cost is the fixed-partition case of our permutation cost.
We allow reveal sets to depend on generated values and obtain matching
depth bounds for a known target. Sampling an unknown target through
oracle queries \citep{anari2024,anari2026,chen2025schedules,chenwang2026}
and generation under circuit restrictions \citep{dlmoptimal2025} impose
different information or computational constraints.
Section~\ref{sec:nonoracle} treats learned prediction error separately,
using standard risk-transfer and concentration tools.
\citet{kim2026puma} identify the mismatch between random training masks
and structured inference masks and propose training aligned with the inference masks. Our training
analysis instead contributes three things: the risk required by the anchor
schedule, the measurement of schedule-specific effects through
the fixed-schedule negative log-likelihood (NLL) identity, and
budget-matched controls.

\section{Reveal Model and Serial Depth}\label{sec:prelim}

This section sets up the objects that the rest of the paper bounds. A
decoder is modeled as a rule that, after each round, sees the values filled
so far and chooses the next batch, which is then sampled independently from
exact conditional distributions. Serial depth is the smallest average number
of rounds such a rule needs to keep the forward KL error within a budget.
Section~\ref{sec:notation} fixes notation and recalls conditional total
correlation, the quantity that measures what independent sampling loses.
Section~\ref{sec:model} defines rules and serial depth and states the
modeling assumptions.

\subsection{Notation and Information Quantities}\label{sec:notation}\label{sec:info}

We collect the notation used throughout and define two information
quantities. Forward KL divergence measures the decoding error. Conditional
total correlation measures the dependence lost when a batch is sampled
independently.
Write $[n]=\{1,\ldots,n\}$ and let $V$ be a finite alphabet. Uppercase
letters denote random variables and lowercase letters their realizations;
$X_A$, $x_A$, $|A|$ and $p_A$ denote subvectors, cardinality and marginals
on $A\subseteq[n]$. We use $X_{\le i}=X_{[i]}$ and
$X_{\ge i}=X_{\{i,\ldots,n\}}$. Conditional laws $p(x_B\mid x_A)$ are
evaluated only for disjoint $A,B$ with $p_A(x_A)>0$. Write
$\mathrm{supp}(p)$ for support, $\mathbf1_E$ for an indicator, and
$\Pr_p,\E_p$ for probability and expectation, omitting $p$ when clear.
All logarithms are base two with $0\log_2 0=0$, so information quantities
are in bits. We write $\ln$ for the natural logarithm. We write $H,I,\KL{p}{q}$,
$\TC$ and $\tc$ for entropy, mutual information, forward KL divergence,
total correlation and its realized value. A realized context is written
as $H(X_A\mid x_C)$ or $\TC(X_A\mid x_C)$. Finally,
$a_+=\max\{a,0\}$, $a\wedge b=\min\{a,b\}$ and $a_n\asymp b_n$ means
$a_n=\Theta(b_n)$. Asymptotics are in $n$; subscripts list allowed
parameter dependencies, and constants are positive and independent of
$n$ unless stated otherwise.

Let $X\sim p$ on $V^n$, without assuming full support. We use standard
entropy and mutual information notation \citep{covertho2006}. Forward KL
divergence is
\[
\KL{p}{q}=\sum_{x:p(x)>0}p(x)\log\frac{p(x)}{q(x)},
\]
with value $+\infty$ if $p(x)>0=q(x)$ for some $x$. It is the excess
expected logarithmic loss under $q$ \citep{merhavfeder1998}. For $A\cap C=\emptyset$ and $p_C(x_C)>0$, the total correlation
\citep{watanabe1960} of the conditional law $p(\cdot\mid x_C)$, also
called multi-information, is
\[
\TC(X_A\mid x_C)=\sum_{i\in A}H(X_i\mid x_C)-H(X_A\mid x_C),
\qquad \TC(X_\emptyset\mid x_C)=0.
\]

\begin{lemma}[Basic properties of conditional total correlation]
\label{lem:tc-basic}\label{lem:product-additivity}
Fix $C\subseteq[n]$ and $p_C(x_C)>0$.
\begin{enumerate}[label=(\roman*),leftmargin=2.2em]
\item For $A\subseteq[n]\setminus C$, $\TC(X_A\mid x_C)\ge0$, with equality
exactly when $p(x_A\mid x_C)=\prod_{i\in A}p(x_i\mid x_C)$.
\item For disjoint $A_1,\ldots,A_k\subseteq S\subseteq[n]\setminus C$,
$\TC(X_S\mid x_C)\ge\sum_l\TC(X_{A_l}\mid x_C)$. In particular,
$\TC(X_S\mid x_C)\ge I(X_i;X_j\mid x_C)$ for distinct $i,j\in S$.
\item If $p(x)=\prod_{a=1}^r p^{(a)}(x_{\mathcal A_a})$ for a partition
$\{\mathcal A_a\}_{a=1}^r$ of $[n]$, then conditioning on $x_C$ preserves
this product across the remaining sets $\mathcal A_a\setminus C$, with
factor $p^{(a)}(x_{\mathcal A_a\setminus C}\mid x_{\mathcal A_a\cap C})$, and
$\TC(X_S\mid x_C)=\sum_{a=1}^r\TC(X_{S\cap\mathcal A_a}\mid x_C)$ for
every $S\subseteq[n]\setminus C$.
\end{enumerate}
\end{lemma}

Part~(i) is the independence bound \citep{covertho2006} applied to
$p(\cdot\mid x_C)$, and parts~(ii) and~(iii) are instances of the
decomposition of total correlation over nested partitions
\citep{watanabe1960}. The three facts follow, respectively, by writing $\TC$ as KL to the product of marginals; iterating
$\TC(X_{A\cup B}\mid x_C)=\TC(X_A\mid x_C)+\TC(X_B\mid x_C)
+I(X_A;X_B\mid x_C)$ for disjoint $A,B$; and dividing the product law by
$p_C(x_C)$.

\subsection{Adaptive Reveal Model}\label{sec:model}

We now formalize decoding rules and define serial depth.

\begin{definition}[Adaptive reveal rule]\label{def:policy}
A \emph{reveal state} is $(C,x_C)$, where $C\subseteq[n]$ and $x_C\in V^C$.
A \emph{deterministic value-adaptive rule} $\pi$ maps each state with
$C\ne[n]$ to a nonempty set $\pi(C,x_C)\subseteq[n]\setminus C$.
\end{definition}

A \emph{fixed schedule} chooses each round's reveal set from the round
index alone. A rule together with its sampling conditionals is a
\emph{decoder}. Starting from $C_1=\emptyset$, round $t$ chooses
$S_t=\pi(C_t,x_{C_t})$, draws $x_i\sim p(\cdot\mid x_{C_t})$
independently for $i\in S_t$, and sets $C_{t+1}=C_t\cup S_t$. Write
$q_\pi$ for the output law and $R(x)$ for its round count, setting
$C_t=[n]$ after termination. Exact conditional marginals isolate the cost
of within-round independence from prediction error. Without the product
constraint, auxiliary randomness can encode arbitrary joint dependence in
one call \citep{draxler2026}. \citet[Proposition~1]{copula2024}
lower-bound the negative evidence lower bound of fully factorized
denoising by dual total correlations, and models with within-round
dependence are proposed and evaluated in
\citet{copula2024,factbarrier2026}.

If the process reaches a \emph{null history} with $p_C(x_C)=0$, the oracle
returns a fixed law, say uniform on $V$, so $q_\pi$ remains well defined.
For $p(x)>0$, every visited marginal satisfies $p_{C_t}(x_{C_t})\ge p(x)>0$;
thus this convention affects neither the factorization in
Proposition~\ref{thm:identity} nor $\KL{p}{q_\pi}$. A state $(C,x_C)$ is
\emph{reachable} if the trajectory of some $x$ with $p(x)>0$ passes
through it; then $p_C(x_C)>0$.

\begin{definition}[Serial depth]\label{def:depth}
For $\eps\ge0$, let $\Pi_\eps(p)$ contain the deterministic value-adaptive
rules with $\KL{p}{q_\pi}\le\eps$. The \emph{serial depth} and its
\emph{hard-cap} variant are
\[
D_\eps(p):=\min_{\pi\in\Pi_\eps(p)}\E_p[R],\qquad
\overline D_\eps(p):=\min_{\pi\in\Pi_\eps(p)}\max_{x\in V^n}R(x).
\]
\end{definition}

There are finitely many rules, and revealing one coordinate per round
is exact. Hence both minima exist, $D_\eps(p)\le\overline D_\eps(p)\le n$,
and both depths are nonincreasing in $\eps$.

\begin{remark}[Target-averaged and runtime round counts]\label{rem:whichmean}
The target mean $\E_p[R]$ and runtime mean $\E_{q_\pi}[R]$ can differ for
an adaptive rule; both are bounded by $\max_xR(x)$. Except for the
$D_\eps$ bounds in Theorem~\ref{thm:binarynear}, every depth upper bound
in this paper uses a fixed schedule, whose constant round count controls
all three quantities. Runtime statements for adaptive samplers should use
$\overline D_\eps$ or $\E_{q_\pi}[R]$.
\end{remark}

\section{The Exact Cost Identity}\label{sec:main}

Serial depth is a minimum over all rules, so bounding it requires a way to
evaluate the error of any rule. This section supplies that tool and shows
how it yields depth bounds. Section~\ref{sec:identity} states the identity
that writes the error of a deterministic rule as the expected sum of the conditional total correlations of its batches. It then explains how
lower bounds on the dependence lost in each round become lower bounds on
depth, and how schedules attain them. Section~\ref{sec:examples} shows that depth
is not determined by likelihood or entropy. The extension to learned
conditionals is deferred to Section~\ref{sec:learned}, where it is used.

\subsection{Exact Information Cost}\label{sec:identity}

This subsection states the cost identity and shows how it turns into depth
bounds.
For disjoint $S,C$, define the \emph{realized total correlation}
\[
\tc(x_S\mid x_C):=\log\frac{p(x_S\mid x_C)}
{\prod_{i\in S}p(x_i\mid x_C)},
\]
whose conditional expectation under $p$ is $\TC(X_S\mid x_C)$.

\begin{proposition}[Adaptive cost identity; cf.\ \citealp{entropyunmask2025}]\label{thm:identity}
For a deterministic value-adaptive rule $\pi$ and $p(x)>0$, the unique
trajectory $(C_t(x),S_t(x))_{t\le R(x)}$ has positive-probability contexts
and satisfies
\begin{equation}\label{eq:chain}
\begin{aligned}
p(x)&=\prod_{t=1}^{R(x)}p\bigl(x_{S_t(x)}\mid x_{C_t(x)}\bigr),\\
q_\pi(x)&=\prod_{t=1}^{R(x)}\prod_{i\in S_t(x)}
p\bigl(x_i\mid x_{C_t(x)}\bigr).
\end{aligned}
\end{equation}
Consequently,
\begin{equation}\label{eq:identity}
\begin{aligned}
\KL{p}{q_\pi}
&=\E_{x\sim p}\sum_{t=1}^{R(x)}\tc\bigl(x_{S_t(x)}\mid x_{C_t(x)}\bigr)\\
&=\E_{x\sim p}\sum_{t=1}^{R(x)}\TC\bigl(X_{S_t}\mid x_{C_t}\bigr)\ge0.
\end{aligned}
\end{equation}
On $\{t\le R\}$, $\E_p[\tc(X_{S_t}\mid X_{C_t})\mid C_t,X_{C_t}]
=\TC(X_{S_t}\mid X_{C_t})$. Thus $\KL{p}{q_\pi}=0$ exactly when every
round has zero conditional total correlation on every $p$-supported
trajectory.
\end{proposition}

\begin{proof}
This is the deterministic, exact-conditional case of
\citet[Eq.~(7), Appendix~A.1]{entropyunmask2025}. We include a short
argument that fixes the null-history convention needed when $p$ lacks full
support and records the pathwise factorization~\eqref{eq:chain}. For $p(x)>0$, determinism and
nonempty reveal sets give a unique trajectory with $R(x)\le n$; the ratios
$p_{C_{t+1}}(x_{C_{t+1}})/p_{C_t}(x_{C_t})=p(x_{S_t}\mid x_{C_t})$
telescope to $p(x)$, and the
sampling rule gives the second line of~\eqref{eq:chain}. At a reachable
state $(C,x_C)$, all earlier choices use only values in $x_C$, so reaching
that state is the cylinder event $\{X_C=x_C\}$. Hence conditioning on the
state introduces no further selection bias. Taking logarithms and then conditional expectations
gives~\eqref{eq:identity}. Nonnegativity in Lemma~\ref{lem:tc-basic}(i)
gives the zero-cost criterion.
\end{proof}

\begin{remark}[Randomized rules]\label{rem:randomized}
With an independent seed $\omega$, $q_\pi=\E_\omega q_{\pi_\omega}$ and
$\KL{p}{q_\pi}\le\E_\omega\KL{p}{q_{\pi_\omega}}$. This can improve on
deterministic depth: for $p_{\mathrm{perm}}$ of Section~\ref{sec:perm},
mix a one-round schedule with probability $w=1-2^{-\eps}$ and a one-position
schedule otherwise. Then $q_\pi\ge(1-w)p$, so $\KL{p}{q_\pi}\le\eps$ with
$w+(1-w)n=2^{-\eps}n+O(1)$ expected rounds, below
$r_*(\eps)n=\frac{1+2^{-2\eps}}2n$ for $0<\eps\le\frac12\log_2 3$
(Theorem~\ref{thm:frontier}). Our depth bounds therefore concern
deterministic rules.
\end{remark}

\paragraph{From the identity to depth bounds.}
Every term in~\eqref{eq:identity} is nonnegative. A rule therefore meets
the budget $\eps$ only if the dependence it loses over all its rounds is
at most $\eps$ on average over target sequences. A lower bound on depth
therefore follows from a lower bound on the loss per round. By
Lemma~\ref{lem:tc-basic}(ii), a batch loses at least the sum of the
conditional mutual informations $I(X_i;X_j\mid x_C)$ of any disjoint pairs
of its positions, given the revealed context. Few rounds force large
batches. Suppose that every large batch must contain many disjoint dependent pairs
at every state a rule can reach. Then the accumulated loss exceeds $\eps$
whatever positions the rule chooses after seeing the values.
Conversely, a schedule attains a matching upper bound by filling in each
round only positions that are nearly independent given what has been
revealed. The copy chain of Example~\ref{ex:copy} is the simplest
instance, and Section~\ref{sec:families} carries out this program for
structured families.

\subsection{Depth versus Likelihood}\label{sec:examples}

Serial depth is not a function of likelihood or entropy, as the following
example shows.

\begin{example}[Likelihood does not determine depth]
\label{ex:indep}\label{ex:copy}\label{ex:nll}
Let $1\le k<n$, let $Y_1,\ldots,Y_k$ be independent fair bits, and set
$\tau(j)=1+((j-1)\bmod k)$.
\begin{enumerate}[label=(\roman*),leftmargin=2.2em]
\item The law $p_{\mathrm{ind}}$ of $(Y_1,\ldots,Y_k,0,\ldots,0)$ has
independent coordinates, so $D_0(p_{\mathrm{ind}})=1$.
\item The \emph{copy chain} $p_{\mathrm{cop}}$, the law of
$(Y_1,\ldots,Y_k,Y_{\tau(1)},\ldots,Y_{\tau(n-k)})$, has
$D_0(p_{\mathrm{cop}})=2$: reveal $[k]$, then its deterministic copies.
A one-round decoder has cost $\TC(X_{[n]})=n-k>0$ because every coordinate
is a fair bit and $H(X_{[n]})=k$. Since $S_1$ is fixed, any rule with $\E_p[R]<2$ must reveal everything in one round.
\item The uniform law $p_{\mathrm{cube}}$ on $\{0,1\}^n$ has
$-\log_2p_{\mathrm{cube}}(x)=H(p_{\mathrm{cube}})=n$, the maximal entropy,
but $D_0(p_{\mathrm{cube}})=1$.
\end{enumerate}
Both $p_{\mathrm{ind}}$ and $p_{\mathrm{cop}}$ have $-\log_2p(X)=k$ almost
surely, yet their depths differ.
\end{example}

Part (iii) also shows why the likelihood-based round lower bound of
\citet[Theorem~3.2]{fu2025bits}, which is pathwise and holds under a
batch-size-dependent confidence constraint on the committed values,
concerns that constrained decoder class rather than all reveal rules.

\section{Serial Depth of Structured Families}\label{sec:families}

By the cost identity, a lower bound on depth follows from a lower bound on
the dependence that any batch of a given size loses at every state a rule
can reach. This section carries out that program for
families chosen to span local and global dependence and constructs
schedules that match these bounds (Table~\ref{tab:depth-regimes}).
Section~\ref{sec:local} treats local dependence in finite-order Markov
laws, the Bernoulli walk and mixing binary chains. For the walk, the
reveal order alone separates logarithmic from linear depth. For mixing
chains, a positive error budget lowers depth to doubly logarithmic, and
confidence-ordered selection is far from optimal. Section~\ref{sec:global}
treats global constraints in permutations, balanced strings and one-hot
blocks. These families need from polylogarithmically to linearly many
rounds. Section~\ref{sec:spectrum}
shows that intermediate growth rates also occur for binary laws with full
support. Proofs are in Appendices~\ref{app:markov}--\ref{app:generic}.

\subsection{Local Dependence}\label{sec:local}

This subsection asks how fast decoding can be when dependence is local
along the sequence. Revealed positions cut such dependence. The unrevealed positions therefore
split into stretches that can be completed in parallel. The lower bounds
show that no rule avoids paying for the dependence inside a stretch. For
the Bernoulli walk, each long unrevealed stretch whose values are not
already determined is a random bridge. Positions in its middle part remain
dependent by a constant amount (Lemma~\ref{lem:bulkpair}). A rule must
therefore keep splitting long stretches, which forces logarithmically many
rounds. For mixing chains, consider an interval of $2^r$ positions that is
touched by at most $r$ rounds. It contains two positions filled in one
round with nothing revealed between them (Lemma~\ref{lem:coloring}). Such
a pair at distance $d$ loses at least a constant times $\rho^{2d}$ bits
(Lemma~\ref{lem:chainpair}). With too few rounds, these losses add up
beyond $\eps$.

\paragraph{Design principle.} These lower bounds share one lesson that also guides the schedules. The
loss of a round comes from positions that remain dependent given the
revealed values. Most of it comes from pairs with no revealed position
between them. A good schedule therefore first fills positions
that cut the sequence into nearly independent stretches and then completes
the stretches in parallel. The anchor schedule of Theorem~\ref{thm:mixing}
follows this principle, whereas confidence-ordered selection violates it
(Theorem~\ref{thm:confchain}).

\paragraph{Bounded-order Markov distributions.}\label{sec:markov}

The simplest question is when dependence can be removed exactly.
Dependence of bounded range can be cut by revealing short separating
blocks, which gives logarithmic depth at zero error.

\begin{definition}[Finite-order Markov distribution]\label{def:markov}
A distribution $p$ on $V^n$ is \emph{order-$m$ Markov} if there are
nonnegative potentials $\phi_i:V^{\{i,\ldots,i+m\}}\to[0,\infty)$,
$1\le i\le n-m$, such that
\[
p(x)=\prod_{i=1}^{n-m}\phi_i\bigl(x_{\{i,\ldots,i+m\}}\bigr)
\qquad\text{for all }x\in V^n.
\]
\end{definition}

\begin{proposition}[Hierarchical schedules for finite-order Markov distributions]\label{thm:markov}
Let $1\le m<n$ and let $p$ be order-$m$ Markov on $V^n$. Then
\[
D_0(p)\;\le\; m\bigl(\lceil \log_2 (n/m)\rceil + 1\bigr),
\]
and the bound is witnessed by a fixed schedule, recursive bisection by
$m$-blocks.
\end{proposition}

\citet[Lemma~4]{bansal2026} prove that bisection by middle blocks of $k$
positions (their Algorithm~1) samples an order-$k$ random walk exactly in
parallel depth $O(k\log(n/k))$ when the conditionals are exact.
Proposition~\ref{thm:markov} applies the same schedule to arbitrary
order-$m$ potentials, which need not be time-homogeneous, and gives the
explicit count $m(\lceil\log_2(n/m)\rceil+1)$;
Appendix~\ref{app:markovproofs} gives the proof.

\paragraph{Bernoulli walks and reveal-order separation.}\label{sec:ordergap}

Logarithmic depth depends on the order of revealing. For the Bernoulli
walk, revealing from left to right needs linearly many rounds. Revealing
from the middle needs only logarithmically many, and no rule does better.
Let $p_{\mathrm{walk}}$ be the law of the Bernoulli walk starting at
$x_0=0$, with independent Bernoulli$(1/2)$ increments
$z_i:=x_i-x_{i-1}$, $i=1,\dots,n$. The initial value $x_0$ is not one of
the $n$ sampled coordinates, and the walk takes values in
$\{0,\ldots,n\}$. The walk is order-1 Markov with
per-coordinate conditional entropy $1$ bit. A \emph{contiguous scan} is a
fixed schedule with $R$ rounds whose round-$t$ set consists of the next
$B_t$ unrevealed coordinates in left-to-right order. Let $\Bin(j,\frac12)$
denote the binomial law with $j$ trials and success probability $\frac12$,
and let $c_0:=1+\tfrac12\log_2(e/\uppi)<0.9$.

\begin{proposition}[Order gap]\label{thm:ordergap}
For the target $p_{\mathrm{walk}}$:
\begin{enumerate}[label=(\roman*),leftmargin=2.2em]
\item The bisection schedule of Proposition~\ref{thm:markov} with $m=1$ uses
$\lceil\log_2(n+1)\rceil$ rounds, and its output law $q$ satisfies
$\KL{p_{\mathrm{walk}}}{q}=0$.
\item Every contiguous scan $\pi$ satisfies
\[
\KL{p_{\mathrm{walk}}}{q_\pi}
=\sum_{t=1}^{R}\sum_{j=1}^{B_t}H\bigl(\Bin(j,\tfrac12)\bigr)-n
\ge\frac n2\log_2\frac nR-c_0n,
\]
and $\KL{p_{\mathrm{walk}}}{q_\pi}\le\eps$ implies
$R\ge n\,2^{-2(\eps/n+c_0)}$.
\end{enumerate}
\end{proposition}

For $\eps=O(n)$, contiguous scans require $\Omega(n)$ rounds; this includes block diffusion \citep{blockdiff2025} run with a single
denoising step per block.
The next theorem allows every deterministic value-adaptive rule, including
ones that clear degenerate segments for free. The longest constant run of
increments accounts for this saving at zero error.

\begin{theorem}[Depth of the Bernoulli walk]\label{thm:walkconverse}\label{thm:walkfixed}
For every deterministic value-adaptive rule $\pi$:
\begin{enumerate}[label=(\roman*),leftmargin=2.2em]
\item If $\KL{p_{\mathrm{walk}}}{q_\pi}=0$, then, with $\mathrm{maxrun}(x)$
the length of the longest constant run of increments of $x$,
\begin{align*}
\E_p[R]
&\ge\log_2(n+1)-\E_p\!\left[\log_2(1+\mathrm{maxrun}(X))\right]\\
&\ge\log_2 n-\log_2\log_2 n-O(1),\qquad n\ge2.
\end{align*}
\item There are absolute constants $c_2,c_3$ such that, for every
$\eps\ge0$, if $\KL{p_{\mathrm{walk}}}{q_\pi}\le\eps$, then
\[
\E_p[R]\;\ge\;\tfrac13\log_2 n-\tfrac13\log_2\log_2 n-c_2-c_3\,\eps
\qquad(n\ge2).
\]
\end{enumerate}
\end{theorem}

Together with Proposition~\ref{thm:ordergap}(i), these bounds give
$D_\eps(p_{\mathrm{walk}})=\Theta_\eps(\log n)$ for fixed $\eps\ge0$, with
an $O(\log\log n)$ gap from bisection at zero error. The proof of~(ii) gives
$c_3\approx1.25\times10^9$, so its conclusion is asymptotic; part~(i)
provides the useful explicit zero-error bound.

\paragraph{Mixing binary chains.}\label{sec:mixing}

Can a small error budget do better than logarithmic depth? For the
Bernoulli walk it cannot, because distant positions remain strongly
dependent. When dependence decays with distance, it can. For $0<\rho<1$,
let
$p_\rho$ be the law of the stationary chain on $\{-1,1\}^n$ with $X_1$
uniform and $p_\rho(X_{i+1}=X_i\mid X_i)=\frac{1+\rho}2$, so that
$\E[X_iX_j]=\rho^{|i-j|}$.

\begin{theorem}[Mixing binary chains]\label{thm:mixing}
For every fixed $0<\rho<1$ and $\eps>0$,
\[
D_\eps(p_\rho)=\log_2\log_2n+O_{\rho,\eps}(1)
\quad\text{and}\quad
\overline D_\eps(p_\rho)=\log_2\log_2n+O_{\rho,\eps}(1),
\]
whereas $D_0(p_\rho)=\overline D_0(p_\rho)=\lceil\log_2(n+1)\rceil$.
\end{theorem}

A fixed schedule attains the upper bound. It first fills every $\ell$th
position with $\ell\asymp\log n$. These anchors are nearly independent,
because the dependence between positions $\ell$ apart decays like
$\rho^{2\ell}$, so filling them together costs little error. Each
remaining gap then has length about $\log n$ and, given its two endpoints,
is independent of the other gaps. Bisecting all gaps in parallel therefore
finishes in about $\log_2\log_2n$ further rounds and loses no further
dependence. At zero error no dependence may be lost, and bisecting the
whole sequence needs $\lceil\log_2(n+1)\rceil$ rounds. At $n=1024$ and
$\rho=0.9$, six rounds with spacing 32 lose $0.026$ bits, whereas random
and confidence selection with six rounds lose $93.8$ and $159.4$ bits.
Appendix~\ref{app:mixing} gives the proof and the details of this
comparison. Confidence selection with ties broken from the left cannot
spread its rounds in this way, as the next result shows. Write
$h_2(u)=-u\log_2u-(1-u)\log_2(1-u)$.

\begin{theorem}[Confidence selection on mixing chains]\label{thm:confchain}
Let $0<\rho<1$ and
$\gamma_\rho:=h_2\bigl(\tfrac{1+\rho^2}2\bigr)-h_2\bigl(\tfrac{1+\rho}2\bigr)$.
At a state $(C,x_C)$, order the positions outside $C$ by decreasing
confidence $\max_{v}p_\rho(X_i=v\mid x_C)$, breaking ties in favor of
smaller indices. Let $\pi$ be a deterministic value-adaptive rule that
at every state reveals the first $B(C,x_C)\ge1$ positions in this order,
where the number $B(C,x_C)$ may depend on the state. Then
$\gamma_\rho>0$; on every trajectory, each $C_t$ is a prefix
$\{1,\dots,k\}$ and $S_t$ is the interval that follows it; and
\[
\KL{p_\rho}{q_\pi}\ge\gamma_\rho\,\E_p\bigl[n-R\bigr],
\]
with equality if and only if, on every trajectory, $|S_1|=1$ and
$|S_t|\le2$ for $t\ge2$. Consequently, $\KL{p_\rho}{q_\pi}\le\eps$ implies
$\E_p[R]\ge n-\eps/\gamma_\rho$.
\end{theorem}

Since $h_2$ is strictly decreasing on $[\frac12,1]$, minimum-entropy selection gives the same order. Threshold rules also
select prefixes of it, including rules with state-dependent thresholds and
fallback to the most confident position. All of these rules therefore need at least
$n-\eps/\gamma_\rho$ rounds on average, versus
$\log_2\log_2n+O_{\rho,\eps}(1)$ for anchors; $\gamma_{0.9}\approx0.167$
bits. Along the trajectories of these rules, ties occur only in the first round,
in which every marginal is uniform. A revealed prompt to the left of the
chain induces the same prefix order without ties. The tie rule cannot be
dropped: starting at $\lceil n/2\rceil$ and expanding both ends samples
exactly in $\lfloor n/2\rfloor+1$ rounds. For learned conditionals, the
bound transfers through Corollary~\ref{cor:learnedub} only to decoders whose
confidences select the same prefix sets on every trajectory, not to
arbitrary learned confidence rankings.

For mixing chains, Theorem~\ref{thm:confchain} gives a lower bound of the
kind asked for in \citet[Section~5]{cai2026confidence}, for deterministic
confidence-ordered rules and with rounds averaged over target sequences.
It does not cover their random-order entropy-sum rule.
\citet[Theorem~3.2]{fu2025bits} prove pathwise round lower bounds under a
batch-size-dependent confidence constraint on the committed values, and
for every fixed confidence factor $f<1$ these bounds are also linear on
$p_\rho$. Theorem~\ref{thm:confchain}
instead bounds sampled outputs under a
forward-KL budget for any state-dependent batch sizes. The proof is in
Appendix~\ref{app:mixing}.

\subsection{Global Constraints}\label{sec:global}

This subsection asks whether exact conditionals always permit fast
decoding. Under constraints that couple all positions they do not, and the
three families need from polylogarithmically to linearly many rounds. For permutations and balanced strings, the unrevealed positions remain
exchangeable given any revealed values. The loss of a batch then depends
only on its size and on a few counts of what remains
(Lemmas~\ref{lem:split} and~\ref{lem:balcost}), so the lower bounds reduce
to bounds on batch sizes. For one-hot blocks, a rule cannot locate
the one in a block faster than about one candidate per round without
paying a loss (Lemmas~\ref{lem:hotcost} and~\ref{lem:hotfree}).

\paragraph{Uniform random permutations.}\label{sec:perm}

Is fast decoding always possible once the conditionals are exact? Uniform
permutations show that it is not: every revealed value constrains all
remaining positions, and depth is linear even with a positive error
budget.
Let $p_{\mathrm{perm}}$ be uniform on permutations of $[n]$, with alphabet
$V=[n]$. Given $M$ unrevealed positions, a batch of size $B$ has cost
\[
g(M,B):=\log_2\frac{M^B}{(M)_B},\qquad (M)_B=M(M-1)\cdots(M-B+1).
\]
Indeed, $2^{-g(M,B)}$ is the probability that $B$ independent uniform draws
from $M$ values are distinct. Exchangeability gives
$\tc(x_S\mid x_C)=\TC(X_S\mid x_C)=g(M,B)$ at every supported state, hence
\begin{equation}\label{eq:permflat}
\sum_{t=1}^{R(x)}\tc(x_{S_t}\mid x_{C_t})
=\sum_{t=1}^{R(x)}g(M_t,B_t),
\end{equation}
where $M_t$ and $B_t$ are the remaining and selected counts. For fixed
partitions, this cost appears in \citet{parallelbench2026}.

\begin{theorem}[Linear lower bounds for permutations]\label{thm:perm}\label{prop:permlog}
Let $\pi$ be a deterministic value-adaptive rule on $p_{\mathrm{perm}}$.
\begin{enumerate}[label=(\roman*),leftmargin=2.2em]
\item For every $x\in\mathrm{supp}(p_{\mathrm{perm}})$, the realized
trajectory cost~\eqref{eq:permflat} satisfies
$\sum_{t=1}^{R(x)}g(M_t,B_t)\ge\log_2\bigl(n/R(x)\bigr)$.
\item If $\KL{p_{\mathrm{perm}}}{q_\pi}\le\eps$ for some $\eps\ge0$, then
\[
\E_{p_{\mathrm{perm}}}[R]\;\ge\;n\max\Bigl\{\frac{1}{1+2\eps\ln2},\,2^{-\eps}\Bigr\}.
\]
\end{enumerate}
Consequently
$D_\eps(p_{\mathrm{perm}})\ge n\max\{1/(1+2\eps\ln2),2^{-\eps}\}$ for every
$\eps\ge0$.
\end{theorem}

The constant in this linear bound is determined exactly. Under a hard cap
of $R$ rounds, the minimum error over all rules is attained by a fixed
schedule and reduces to an optimization over batch sizes
(Theorem~\ref{thm:hardcap} in Appendix~\ref{app:perm}); its limit gives
the following frontier.

\begin{theorem}[Joint-scaling frontier]\label{thm:frontier}
For every fixed $\eps\ge0$,
\[
\lim_{n\to\infty}\frac{D_\eps(p_{\mathrm{perm}})}{n}
\;=\;\lim_{n\to\infty}\frac{\overline D_\eps(p_{\mathrm{perm}})}{n}
\;=\;r_*(\eps),
\]
where $r_*$ is the nonincreasing function defined
in~\eqref{eq:rstar} of Appendix~\ref{app:perm}. In particular,
$r_*(\eps)=(1+2^{-2\eps})/2$ for $0\le\eps\le\frac12\log_23$, so that
$r_*(\frac12)=\frac34$: half a bit of total error still requires about
three rounds for every four positions.
\end{theorem}

Appendix~\ref{app:perm} proves the results on permutations.
Figure~\ref{fig:permtrade} shows the hard-cap tradeoff at $n=240$.

\begin{figure}[t]
\centering
\includegraphics[width=\linewidth]{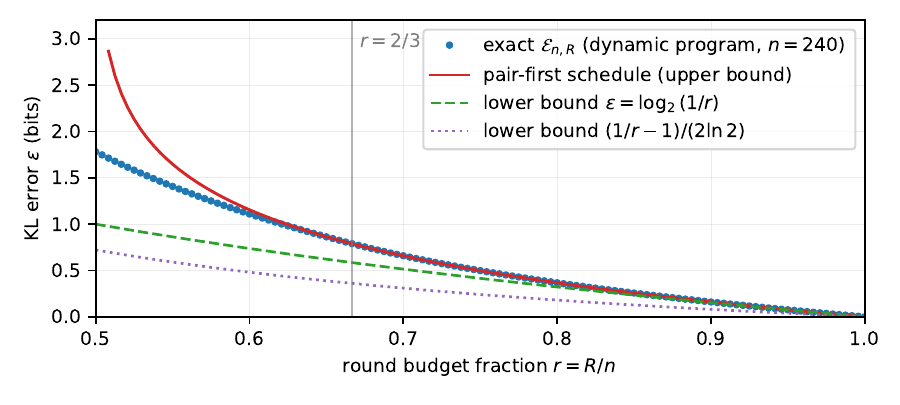}
\caption{Permutation error--round tradeoff at $n=240$. Dots: exact minimum error $\cE_{n,R}$ under a hard cap of $R$ rounds, from the dynamic program. Solid: pair-first schedule, exact for
$n-R\le n/4$ (Proposition~\ref{thm:phase1}). Dashed/dotted: the all-rule
lower bounds $\eps\ge\log_2(1/r)$ and $(1/r-1)/(2\ln2)$, where
$r=\E_p[R]/n$. The limiting optimal profile first uses triples below $r=2/3$.}

\label{fig:permtrade}
\end{figure}

\paragraph{Balanced binary strings.}\label{sec:balanced}

Large alphabets are not needed for slow decoding: a single global
constraint on the number of ones in a binary string already forces more
than logarithmic depth. A batch that is a fraction $\gamma$ of the
unrevealed positions loses about $\gamma^2$ bits, and about $(\log
n)/\gamma$ such rounds are needed, so keeping the total loss within $\eps$
forces order $\log^2n$ rounds.
For $n$ even, let $p_{\mathrm{bal}}$ be uniform on the balanced
binary strings $\{x\in\{0,1\}^n:\sum_ix_i=n/2\}$; equivalently, the
increment sequence of the walk of Section~\ref{sec:ordergap} conditioned
to end at $n/2$.

\begin{theorem}[Depth of balanced strings]\label{thm:balanced}
There are absolute constants $c,c',C_\star>0$ such that, with
$L_n(\eps):=\bigl(\log_2 n-C_\star(\log_2\log_2 n+\eps+1)\bigr)_+$, every
even $n$ and every $0<\eps\le\log_2 n$ satisfy
\[
c\min\Bigl\{\log_2^3n,\frac{L_n(\eps)^2}{\eps}\Bigr\}
\le D_\eps(p_{\mathrm{bal}})
\le\overline D_\eps(p_{\mathrm{bal}})
\le C_\star\Bigl(\frac{\log_2^2n}{\eps}+\log_2 n\Bigr).
\]
Consequently, $D_\eps(p_{\mathrm{bal}})=\Theta_\eps(\log^2n)$ for every fixed
$\eps>0$, and $D_\eps(p_{\mathrm{bal}})=\Theta(\eps^{-1}\log^2n)$ uniformly
over $1/\log_2n\le\eps\le c'\log_2n$.
\end{theorem}

The lower bound uses an exchangeable urn coupling, under which the
conditional law depends only on how many positions are selected. A
fixed fractional-batch schedule gives the upper bound
(Appendix~\ref{app:balanced}).

\paragraph{Binary one-hot blocks.}\label{sec:onehot}

Many small exact constraints, one per block, force polynomial depth on a
binary alphabet.
Partition $n=bL$ binary coordinates into $b$ independent blocks, each uniform
on its $L$ one-hot vectors: $X_{a,j}=\mathbf1\{J_a=j\}$ with independent
$J_a\sim\mathrm{Unif}[L]$. Denote this law by $p_{b,L}^{\mathrm{hot}}$ and put
$p_m^{\mathrm{hot}}=p_{m,m}^{\mathrm{hot}}$.

\begin{theorem}[Binary one-hot depth]\label{thm:onehot}\label{cor:onehotrect}
There is an absolute constant $c>0$ such that, for all $b\ge1$, $L\ge16$
and $\eps\ge0$,
\[
c\min\Bigl\{L,\frac{b}{1+\eps}\Bigr\}
\le D_\eps(p_{b,L}^{\mathrm{hot}})
\le\overline D_\eps(p_{b,L}^{\mathrm{hot}})
\le L.
\]
In particular, $D_\eps(p_m^{\mathrm{hot}})=\Theta_\eps(\sqrt n)$ along $n=m^2$
for every fixed $\eps\ge0$, and every exponent $\alpha\in(0,\frac12]$ occurs
as $D_\eps=\Theta_\eps(n^\alpha)$ along families $p_{b,L}^{\mathrm{hot}}$.
\end{theorem}

Deferred decisions give the adaptive lower bound, and revealing one position
per block per round gives the zero-error upper bound
(Appendix~\ref{app:onehot}); the bounds match for $b\ge L$ at fixed $\eps$.

\subsection{Depth Spectrum and Further Directions}\label{sec:spectrum}

Two questions remain. Do global constraints by themselves force many
rounds? Example~\ref{ex:affine} shows that they do not: uniform laws on
affine subspaces need only two rounds, so the local--global grouping above
organizes the examples but does not determine depth. Is near-linear depth
an artifact of large alphabets or of zero-probability configurations?
Theorem~\ref{thm:binarynear} shows that it is not, using binary laws in
which every string has at least a quarter of its uniform probability.
Exactly resolvable constraints keep the zero-error depth small.

\begin{example}[Affine subspaces]\label{ex:affine}
If $p$ is uniform on an affine subspace $\{x\in\mathbb F_q^n:Ax=b\}$ over a
finite field $\mathbb F_q$, then $D_0(p)\le2$. After Gaussian elimination,
the free coordinates are independent and uniform, and each pivot coordinate
is an affine function of them, so revealing the free coordinates in one
round and the pivots in a second round is exact.
\end{example}

In contrast, no $o(n)$ upper bound on depth holds for all distributions.
At zero error this follows from Proposition~\ref{prop:generic} below. At
any fixed $\eps>0$ it follows from Theorem~\ref{thm:perm}, whose alphabet
grows with $n$. Over a binary alphabet the question is open (Open
Problem~\ref{op:adaptive}).

\begin{proposition}[Generic distributions at zero error]\label{prop:generic}
Let $|V|\ge2$. For Lebesgue-almost every $p$ in the interior of the simplex
$\Delta(V^n)$, that is, with full support, every deterministic
value-adaptive rule with $\KL{p}{q_\pi}=0$ satisfies $|S_t(x)|=1$ for all
$x\in V^n$ and $t\le R(x)$. Consequently $D_0(p)=\overline D_0(p)=n$.
\end{proposition}

Generic conditional dependence may be arbitrarily weak, so
Proposition~\ref{prop:generic} concerns zero error only. The next construction
has nearly linear depth at every fixed positive budget.

\begin{theorem}[Near-linear binary depth]\label{thm:binarynear}
There are laws $p_n$ on $\{0,1\}^n$, for $n\ge256$ a power of two, with
$p_n(x)\ge2^{-n}/4$ for all $x$ such that, for every $\eps\ge0$,
\[
\overline D_\eps(p_n)\ge D_\eps(p_n)\ge\frac{n}{128\,(\eps+5\log_2n+1)},
\]
and, for every fixed $\eps>0$, $D_\eps(p_n)\le10n/\log_2n$ for all
sufficiently large $n$. For every $\alpha\in(0,1)$ there are laws $p'_n$ on $\{0,1\}^n$
with $p'_n\ge2^{-n}/4$ and $D_\eps(p'_n)=\Theta_{\alpha,\eps}(n^\alpha)$ for
every fixed $\eps>0$.
\end{theorem}

The law mixes uniform bits with random codebooks of $2^k$ vectors,
$n/4\le k<n/2$. Small coordinate sets appear nearly uniform and conceal when
a component's randomness runs out; crossing many such thresholds incurs cost.
The matching upper bound is an information-theoretic existence result for
target-averaged depth; its rule uses the known codebooks and is not
claimed to be computationally efficient (Appendix~\ref{app:generic}).

\begin{openproblem}[Linear depth over a binary alphabet]\label{op:adaptive}
For a fixed tolerance $\eps>0$, does there exist a family of distributions
$p_n$ on $\{0,1\}^n$ such that $D_\eps(p_n)=\Omega(n)$?
\end{openproblem}

\section{Prediction Error and Fast Decoding}\label{sec:nonoracle}

The bounds of Section~\ref{sec:families} assume exact conditionals,
whereas trained models predict them with error. This section asks when the
guarantees survive, and the answers come in three strengths. Under exact
conditionals, serial depth is determined by the results of
Sections~\ref{sec:main} and~\ref{sec:families}. With a learned predictor,
the depth lower bounds still hold, and a schedule keeps its error
guarantee under an accuracy condition that is sufficient but has to be
met. On pretrained models this condition cannot be checked, and only relative
improvements between decoders are certified.
Section~\ref{sec:learned} adds prediction error to the cost identity.
Section~\ref{sec:anchor-risk} computes the accuracy the anchor schedule
needs and shows that training risk transfers to this accuracy when the
training tasks cover its contexts.
It also shows that a small average error under random masks does not
suffice. Section~\ref{sec:training-transfer} measures changes in prediction error
through held-out likelihood and certifies differences between decoders.
The experiments of Section~\ref{sec:decoders-part} test the anchor
schedule obtained for mixing chains; the other families delimit what is
possible and are not each meant to yield a decoding algorithm. When comparing predictors,
$\pi$ is held fixed; a confidence rule recomputed from another predictor is
a different rule.

\subsection{Learned Conditionals}\label{sec:learned}

This subsection adds prediction error to the cost identity and derives
what it implies for lower and upper bounds.
For learned marginals $\hat p_i(\cdot\mid x_C)$, let $\hat q_\pi$ be the
output law of the same rule with independent within-round draws from $\hat p$. The rule may itself be defined through $\hat p$, as in
confidence decoding. Its decomposition follows
\citet[Eq.~(7)]{entropyunmask2025}:

\begin{corollary}[Learned conditionals]\label{cor:learned}
For every deterministic value-adaptive rule $\pi$,
\begin{equation}\label{eq:learned}
\begin{aligned}
\KL{p}{\hat q_\pi}
=\E_{x\sim p}\Biggl[\sum_{t=1}^{R(x)}\Bigl(
&\TC\bigl(X_{S_t}\mid x_{C_t}\bigr)\\
&{}+\sum_{i\in S_t}
\KL{p(X_i\in\cdot\mid x_{C_t})}{\hat p_i(\cdot\mid x_{C_t})}\Bigr)\Biggr],
\end{aligned}
\end{equation}
allowing the value $+\infty$ on both sides.
\end{corollary}

\begin{proof}
Use the state-conditioning argument in the proof of
Proposition~\ref{thm:identity} and insert the exact marginal in each log
ratio. If a learned marginal vanishes on a positive-probability
outcome at a reachable selected coordinate, both sides are infinite;
otherwise all terms are finite.
\end{proof}

The first term is the \emph{dependence cost} $\KL{p}{q_\pi}$. Write
\[
\zeta_i(x_C):=\KL{p(X_i\in\cdot\mid x_C)}{\hat p_i(\cdot\mid x_C)}
\]
for the \emph{prediction error}. Reveal choices affect both terms.

\begin{corollary}[Lower bounds for learned conditionals]\label{cor:learnedlb}
If a deterministic value-adaptive rule $\pi$, possibly defined through
$\hat p$, satisfies $\KL{p}{\hat q_\pi}\le\eps$, then
$\E_p[R]\ge D_\eps(p)$ and $\max_xR(x)\ge\overline D_\eps(p)$.
\end{corollary}

\begin{proof}
By~\eqref{eq:learned}, \eqref{eq:identity} and $\zeta\ge0$,
$\KL{p}{q_\pi}\le\KL{p}{\hat q_\pi}\le\eps$, where $q_\pi$ uses the same
reveal map and round function with exact conditionals; so
$\pi\in\Pi_\eps(p)$.
\end{proof}

Define the accumulated prediction error
\[
\mathcal M(\hat p,\pi):=\E_{x\sim p}\sum_{t=1}^{R(x)}\sum_{i\in S_t(x)}
\zeta_i\bigl(x_{C_t(x)}\bigr).
\]

\begin{corollary}[Upper bounds for learned conditionals]\label{cor:learnedub}
For every deterministic value-adaptive rule $\pi$,
$\KL{p}{\hat q_\pi}=\KL{p}{q_\pi}+\mathcal M(\hat p,\pi)$. Thus an oracle
bound $\KL{p}{q_\pi}\le\eps_1$ gives
$\KL{p}{\hat q_\pi}\le\eps_1+\mathcal M(\hat p,\pi)$ with the same round
function. If $\zeta_i(x_C)\le\bar\zeta$ everywhere, then
$\mathcal M(\hat p,\pi)\le n\bar\zeta$.
\end{corollary}

\begin{proof}
Apply~\eqref{eq:learned} and note that each coordinate is revealed once.
\end{proof}

\subsection{Accuracy Required on Visited Contexts}\label{sec:anchor-risk}

This subsection turns Corollary~\ref{cor:learnedub} into an explicit
accuracy requirement for the anchor schedule and shows that training risk
transfers to it when the training tasks cover the schedule's contexts. It also shows that a small average error under
random masks does not suffice.
For a fixed schedule $\pi$ with batches $S_t$, put $B_t=|S_t|$. Its
\emph{prediction-task law} $\mu_\pi$ draws $x\sim p$, a round $t$ with
probability $B_t/n$, and a position $i\in S_t$ uniformly; the task is
$(C_t,x_{C_t},i)$. With $\delta_\pi(\hat p)$ the mean of
$\zeta_i(x_{C_t})$ under $\mu_\pi$, we have
$\mathcal M(\hat p,\pi)=n\delta_\pi(\hat p)$.

For the chain $p_\rho$ of Section~\ref{sec:mixing} and an integer
$1\le\ell<n$, let $\pi_\ell$ first reveal $\ell,2\ell,\ldots,m_\ell\ell$,
where $m_\ell:=\lfloor n/\ell\rfloor$, and then bisect every remaining gap
in parallel. It uses at most $1+\lceil\log_2\ell\rceil$ rounds, and
\begin{equation}\label{eq:anchor-risk}
\KL{p}{\hat q_{\pi_\ell}}=\Gamma_\ell+n\delta_{\pi_\ell}(\hat p),\qquad
\Gamma_\ell=(m_\ell-1)\Bigl[1-h_2\Bigl(\frac{1+\rho^\ell}2\Bigr)\Bigr]
\le(m_\ell-1)\rho^{2\ell}.
\end{equation}
Indeed, the anchors form a stationary chain of correlation
$\varrho=\rho^\ell$, so their total correlation is
$m_\ell-[1+(m_\ell-1)h_2(\frac{1+\varrho}2)]=\Gamma_\ell$. The identity
$1-h_2((1+\varrho)/2)=(\ln2)^{-1}\sum_{k\ge1}\varrho^{2k}/[2k(2k-1)]$
has nonnegative coefficients that sum to one, including the factor
$(\ln2)^{-1}$. This gives the bound $\varrho^2$. Later batches choose at most one
point per separated gap, hence cost zero by Lemma~\ref{lem:sep} with $m=1$.

Training can cover these prediction tasks explicitly. Let $\nu$ be a
nonnegative measure on tasks $(C,x_C,i)$ with $\nu\ge\lambda\mu_{\pi_\ell}$
for some $\lambda>0$, and let $\delta_\nu(\hat p):=\int\zeta_i(x_C)\,
d\nu(C,x_C,i)$. Nonnegativity of $\zeta$ gives
\begin{equation}\label{eq:training-coverage}
\KL{p}{\hat q_{\pi_\ell}}\le\Gamma_\ell+\frac n\lambda\,\delta_\nu(\hat p).
\end{equation}
Training only on the tasks of $\pi_\ell$ gives $\lambda=1$, and mixing
them in equal proportion with any other tasks gives $\lambda=\frac12$. For a
total budget $\eps>\Gamma_\ell$, the sufficient condition is
$\delta_\nu(\hat p)\le\lambda(\eps-\Gamma_\ell)/n$, a per-token accuracy
that tightens with $n$. For a value-adaptive rule the same argument
requires coverage of the revealed values as well as of the revealed set,
because the values at a visited set $C$ need not follow $p_C$.

\paragraph{Average error under random masks.}\label{sec:avgloss}

Training controls an average error under random masks. A schedule needs
accuracy on its own contexts instead, and the former does not imply the
latter.
Masked-diffusion training draws a revealed set $C$ whose law, given $|C|=k$,
is uniform over the $k$-subsets of $[n]$, draws $x\sim p$, and scores the
predictions of the unrevealed coordinates. For $0\le k\le n-1$ let
\[
\delta_k(\hat p):=\E_C\,\E_{x\sim p}\Bigl[\frac1{n-k}\sum_{i\notin C}
\zeta_i(x_C)\Bigr],\qquad C\text{ uniform over the $k$-subsets of }[n],
\]
and for $\delta,\bar\zeta>0$ let $\mathcal Q(\delta,\bar\zeta)$ be the set
of predictors with $\delta_k(\hat p)\le\delta$ for every $k$ and
$\zeta_i(x_C)\le\bar\zeta$ whenever $p_C(x_C)>0$.
This is population excess log-loss (conditional KL), not empirical training NLL.

\begin{proposition}[Average prediction error and visited contexts]\label{prop:avgloss}
Let $|V|\ge2$ and $n\ge2$. For every deterministic value-adaptive rule
$\pi$ and all $\delta,\bar\zeta>0$, there is
$\hat p\in\mathcal Q(\delta,\bar\zeta)$ such that
\[
\KL{p}{\hat q_\pi}=\KL{p}{q_\pi}+\bigl(n-|S_1|\bigr)\min\{\bar\zeta,n\delta\},
\]
where $S_1$ is the first reveal set of $\pi$.
\end{proposition}

\begin{proof}
Put $\beta=\min\{\bar\zeta,n\delta\}$. For each conditional law
$p_i:=p(X_i\in\cdot\mid x_C)$ choose $v\in V$ with $p_i(v)<1$; the map
$\theta\mapsto\KL{p_i}{(1-\theta)p_i+\theta\mathbf 1_{\{v\}}}$ is
continuous on $[0,1)$, vanishes at $0$ and tends to $+\infty$ as
$\theta\uparrow1$, so some mixture is at divergence exactly $\beta$. Let
$\hat p$ use this mixture at every reachable state $(C,x_C)$ of $\pi$ with
$|C|\ge1$ and every $i\in\pi(C,x_C)$, and $p_i$ everywhere else. Fix
$1\le k\le n-1$. A trajectory passes through at most one state with
$|C|=k$, and a reachable state is the event $\{X_C=x_C\}$, so the
reachable states of size $k$ have total probability at most one. Since
$|\pi(C,x_C)|\le n-k$ and $\binom nk\ge n$,
\[
\delta_k(\hat p)\le\frac{\beta}{(n-k)\binom nk}\sum_{(C,x_C)}p_C(x_C)\,
|\pi(C,x_C)|\le\frac\beta{\binom nk}\le\frac\beta n\le\delta,
\]
the sum running over reachable states of size $k$; moreover
$\delta_0(\hat p)=0$ and $\zeta\le\beta\le\bar\zeta$, so
$\hat p\in\mathcal Q(\delta,\bar\zeta)$.
By Corollary~\ref{cor:learnedub}, the divergence increases by $\beta$
times the expected number of coordinates revealed after the first round,
which is $n-|S_1|$.
\end{proof}

When $\delta\ge\bar\zeta/n$, the construction incurs $\bar\zeta$ per coordinate
after the first round, even under sequential revelation. For confidence rules, it applies to the map frozen at a predictor $\hat
p_0$. The constructed $\hat p$ may induce a different map, so it does not
prove failure of model-aware selection. Uniform random reveal sets instead follow the training
law, as in \citet[Theorem~1]{licai2025}. Fixed schedules can use the coverage
bound~\eqref{eq:training-coverage}.

\subsection{Measuring and Certifying on Pretrained Models}\label{sec:training-transfer}

On pretrained models prediction error cannot be computed directly. This
subsection measures its changes through held-out likelihood and certifies
differences between decoders.
For a fixed, model-independent schedule $\pi$, let $\mathcal L_{\hat p}(\pi)$ be the
population decoder negative log-likelihood (NLL) per predicted token. For
two predictors $\hat p,\tilde p$ with finite decoder losses, the target
entropy and the dependence cost cancel:
\begin{equation}\label{eq:fixed-schedule-change}
\mathcal L_{\hat p}(\pi)-\mathcal L_{\tilde p}(\pi)=\delta_\pi(\hat p)-\delta_\pi(\tilde p).
\end{equation}
Thus held-out NLL changes estimate prediction-error changes without knowing
the true conditionals. An unbiased objective samples round $t$ with probability
$B_t/n$ and averages its token losses per sequence, then across sequences.
Section~\ref{sec:e4} tests this decoder-aligned training.

\paragraph{A relative certificate.}\label{sec:certification}

For pretrained models the target conditionals are unknown, so absolute KL
cannot be computed, but differences between decoders can be certified.
Suppose that two complete decoders are scored on the same documents with
the same tokenizer, initial context, and number of predicted tokens. The
difference of their expected per-token negative log-likelihoods then
equals the difference
of their conditional forward KL divergences per predicted token, because the unknown target
entropy cancels. Fix the candidate decoders, their references, and the scoring settings
before validation documents are sampled. An empirical Bernstein bound
\citep[Theorem~4]{maurerpontil2009} on per-document score differences then
gives an upper bound $U_j$ on this difference for every candidate $j$
simultaneously with probability at least $1-\alpha$ (Appendix~\ref{app:certificate}). A bound $U_j\le0$
certifies candidate $j$ as no worse than its reference. The guarantee
concerns a smoothed version of the candidate and the declared population;
it does not bound absolute KL or imply an ordering of semantic generation
quality.

\section{Experiments on Masked Language Models}\label{sec:decoders-part}

On mixing chains with exact conditionals, widely spaced anchors incur far
less error than random or confidence-ordered selection with the same
number of rounds (Section~\ref{sec:mixing}). This section tests whether
that advantage survives on pretrained models, where prediction error
enters. Section~\ref{sec:empirical} compares schedules on held-out likelihood and
includes a preregistered confirmation on independent data.
Section~\ref{sec:emptask} checks task accuracy, and Section~\ref{sec:e4}
tests training on a schedule's contexts. The main measure is conditional
decoder negative log-likelihood (NLL), in bits per predicted token: the exact likelihood of an evaluation text given its initial context,
computed as in \citet[Theorem~4.3]{turok2026duel} with the context
revealed from the start. NLL differences between decoders equal differences
of their forward KL divergences per predicted token.

\subsection{Models, Data, and Schedules}\label{sec:empprotocol}

We evaluate LLaDA-8B-Base \citep{llada2025}, Dream-v0-Base-7B
\citep{dream2025}, and Qwen2.5-0.5B \citep{qwen2025} converted to masked
diffusion on WikiText-103 \citep{merity2017}. The five domains are
WikiText, news \citep{hermann2015cnndm}, dialogue \citep{koepf2023oasst},
Python code \citep{husain2019codesearchnet}, and long reports
\citep{huang2021govreport}. Each excerpt has 512 tokens. The first token
is given and the remaining $n=511$ are predicted. The anchor schedule $\mathsf A_5=\pi_{16}$
reveals every sixteenth position and then bisects the gaps, five rounds in
all. Random selection $\mathsf R_5$ uses the same batch sizes, and
$\mathsf R_8$ uses eight nearly equal batches; confidence selection uses the
batch sizes of $\mathsf A_5$ and ranks positions by the largest predicted
probability. Appendix~\ref{app:empsettings} gives sample sizes, batch
sizes, and implementation details.

\subsection{Likelihood and Independent Confirmation}\label{sec:empirical}\label{sec:empdecoders}\label{sec:empanchor}

This subsection compares anchors with random and confidence selection at
equal and reduced round budgets and confirms the main comparison on
independent data.

\begin{table}[t]
\centering\small
\setlength{\tabcolsep}{7pt}
\begin{tabular}{lrrrr}
\hline
& \multicolumn{2}{c}{LLaDA} & \multicolumn{2}{c}{Dream}\\
Domain & $\mathsf A_5-\mathsf R_5$ & $\mathsf A_5-\mathsf R_8$ & $\mathsf A_5-\mathsf R_5$ & $\mathsf A_5-\mathsf R_8$\\
\hline
Code     & $-0.319$ & $-0.176$ & $-0.316$ & $-0.094$\\
Dialogue & $-0.397$ & $-0.406$ & $-0.389$ & $-0.248$\\
News     & $-0.656$ & $-0.191$ & $-0.604$ & $-0.032$\\
Reports  & $-0.589$ & $-0.253$ & $-0.581$ & $-0.165$\\
WikiText & $-0.598$ & $-0.104$ & $-0.603$ & $-0.065$\\
\hline
\end{tabular}
\caption{Paired conditional NLL differences in bits per predicted token;
lower is better, and negative values favor anchors. Each random baseline averages three position
seeds. The $\mathsf A_5-\mathsf R_5$ comparison isolates position
selection by matching every batch size. The $\mathsf A_5-\mathsf R_8$
comparison tests a reduction from eight to five rounds. Paired $95\%$
half-widths of the $\mathsf A_5-\mathsf R_5$ differences are at most $0.03$
bits/token.}
\label{tab:parallel-likelihood}
\end{table}

Anchors improve on $\mathsf R_5$ by $0.32$--$0.66$ bits/token on both
larger models in every domain, and all 30 model--domain--seed comparisons
with $\mathsf R_8$ also favor anchors
(Table~\ref{tab:parallel-likelihood}; absolute NLL in
Table~\ref{tab:abs-nll}). The advantage persists for reports of 512 to
2048 tokens. Confidence selection has $2.1$--$2.7$ bits/token higher NLL
than batch-matched random selection. Sequentializing the same orders
slightly favors confidence instead. A good order thus need not be a good
grouping, as Theorem~\ref{thm:confchain} shows on chains. Pseudo-cost
measurements (Corollary~\ref{cor:seq}) support this mechanism, since
confidence-selected batches contain more adjacent positions and adjacency
predicts the extra cost (Appendix~\ref{app:pseudocost}). Five rounds still
lose $0.65$--$1.10$ bits/token relative to one position per round on the non-WikiText domains. The 0.5B
model also fails on code, where anchors are worse than $\mathsf R_5$.\label{sec:empgen}

\paragraph{Independent confirmation.}\label{sec:empconfirmation}
On 4096 dialogues drawn from the test split of UltraChat-200k
\citep{ding2023ultrachat,tunstall2023zephyr} under a preregistered protocol, we compare $\mathsf A_5$ with $\mathsf
R_8$. Each prediction of $\mathsf A_5$ is mixed with the uniform law at
weight $10^{-6}$. Both certified upper bounds are negative
(Table~\ref{tab:confirmation}), which certifies the relative improvement
for both models. The five-round decoder also has $1.55\times$ and
$1.57\times$ lower median latency.
Appendix~\ref{app:empsettings} gives the protocol.

\begin{table}[t]
\centering\small
\begin{tabular}{lrrrr}
\hline
Model & Mean clipped score $\bar Z_j$ & Certified bound $U_j$ & $\mathsf A_5$ (s) & $\mathsf R_8$ (s)\\
\hline
LLaDA & $-0.437$ & $-0.331$ & $0.307$ & $0.477$\\
Dream & $-0.248$ & $-0.141$ & $0.276$ & $0.431$\\
\hline
\end{tabular}
\caption{Relative certification and generation latency. Mean clipped
scores $\bar Z_j$ and certified upper bounds $U_j$
of~\eqref{eq:relative-certificate} are in bits/token, lower is better. The
one-sided empirical Bernstein bound
\eqref{eq:relative-certificate}, with lower clipping at $-2$, controls both
comparisons jointly at 95\% confidence. Latencies are per-text medians
over 128 dialogues on one NVIDIA RTX PRO 5000 GPU (72\,GB), batch size
one, with interleaved runs.}
\label{tab:confirmation}
\end{table}

\subsection{Task Accuracy}\label{sec:emptask}

Likelihood gains need not translate into task accuracy, so we also
evaluate generation on GSM8K \citep{gsm8k2021}. Both larger models answer the first 200 test problems in 80 rounds.
Bisection, random, and confidence selection each fill blocks of 16
positions in five rounds
(Table~\ref{tab:task}). On LLaDA, bisection exceeds random by $8.3$
percentage points (paired $95\%$ interval $[4.4,12.2]$, Holm-corrected
$p<0.001$); on Dream the single-seed difference is $6.0$ ($[-1.2,13.2]$).
Confidence is worse than both other rules ($p<0.001$), and all remain
below the 256-round reference, so the gain is at equal round counts.
Bisection within blocks closely follows the dilated schedule of
\citet{dus2025}, who also report gains over confidence-based unmasking on
GSM8K.
Appendix~\ref{app:emptask} gives prompts, seeds, greedy decoding, and
HumanEval \citep{humaneval2021} results.

\begin{table}[t]
\centering
\small
\begin{tabular}{lrrrrr}
\hline
model & seeds & bisection & random & confidence & 256-round reference\\
\hline
LLaDA-8B-Base & $3$ & $21.8$ & $13.5$ & $4.7$ & $38.0$\\
Dream-v0-Base-7B & $1$ & $20.0$ & $14.0$ & $2.0$ & $44.5$\\
\hline
\end{tabular}
\caption{Accuracy in percent on the first $200$ GSM8K test problems with
sampling at temperature $1$, higher is better. Bisection, random and
confidence selection fill each block of $16$ positions in $5$ rounds with
the same batch sizes, $80$ rounds in all; the reference fills blocks of
$32$ positions one position per round. LLaDA accuracies are averaged over
three sampling seeds.}
\label{tab:task}
\end{table}

\subsection{Training on the Decoder's Contexts}\label{sec:e4}

Since the theory requires accuracy on the contexts a schedule visits, we
test whether training on those contexts improves the schedule. We continue training the 0.5B model on code with three training seeds and
matched supervision budgets (Table~\ref{tab:training-budget}). Training
uses either ordinary random masks or masks aligned to the contexts of
$\mathsf A_5$ or of $\mathsf R_5$. For a fixed decoder, NLL changes equal
prediction-error changes by~\eqref{eq:fixed-schedule-change}, so the
interaction $\Delta\mathsf A_5-\Delta\mathsf R_5$ isolates the
schedule-specific effect. Anchor alignment gives an interaction of about
$-0.38$ bits/token in every seed, whereas alignment to a random plan gives
none. Against ordinary training of equal length, anchor alignment improves
anchors by $0.096$ but worsens random decoding by $0.266$. Training thus
trades accuracy across contexts rather than improving every decoder.
Aligning training masks with inference masks was proposed by
\citet{kim2026puma}, who follow the model's own confidence-ordered
unmasking to speed up training. Here the masks follow a fixed schedule,
and the alignment targets the accuracy required by
the anchor schedule (Section~\ref{sec:anchor-risk}).
Appendix~\ref{app:empsettings} reports the controls, token budgets, and
transfer to WikiText.

\begin{table}[t]
\centering\small
\setlength{\tabcolsep}{5pt}
\begin{tabular}{lrrrr}
\hline
Training arm & Steps & $\Delta\mathsf A_5$ & $\Delta\mathsf R_5$ & Interaction\\
\hline
Thinned ordinary (64\%) & 1000 & $+0.123$ & $+0.114$ & $+0.008$\\
Anchor aligned & 1558 & $-0.364$ & $+0.014$ & $-0.378$\\
Random aligned & 1558 & $-0.185$ & $-0.184$ & $-0.001$\\
Mixed & 1218 & $-0.235$ & $-0.049$ & $-0.186$\\
Ordinary, longer & 1558 & $-0.268$ & $-0.252$ & $-0.016$\\
\hline
\end{tabular}
\caption{Mean changes across three training seeds relative to ordinary
training for 1000 steps, in bits/token; negative values are improvements. Interaction is
$\Delta\mathsf A_5-\Delta\mathsf R_5$, not the anchor decoder's own improvement.
Ordinary training supervises about $267$ tokens per sequence and aligned
training about $172$; the aligned and mixed arms run $1558$ and $1218$ steps
to approximately match the reference's supervised-token count, and the
thinned ordinary arm keeps each ordinary loss term with probability
$0.64\approx172/267$ and reweights it, matching the aligned arms'
supervision per step at $1000$ steps. Values are rounded separately.}
\label{tab:training-budget}
\end{table}

\section{Discussion}\label{sec:discussion}

We introduced serial depth, the fewest rounds that any deterministic
adaptive rule needs to sample a target within a forward-KL error budget.
For structured families we proved lower bounds that hold for every such
rule and constructed schedules that attain them. At a fixed positive error
budget, depth ranges from doubly logarithmic for mixing chains to linear
for permutations (Table~\ref{tab:depth-regimes}). It is not determined by
entropy or likelihood (Example~\ref{ex:indep}). On mixing chains,
confidence-ordered selection with leftmost tie-breaking needs linearly many
rounds where an anchor schedule needs doubly logarithmically many. The
depth lower bounds remain valid for learned predictors
(Corollary~\ref{cor:learnedlb}), and a schedule keeps its error guarantee
when its prediction error on the contexts it visits is small
(Corollary~\ref{cor:learnedub}). On two pretrained models, the anchor
schedule suggested by the theory attains higher held-out likelihood in five
rounds than random selection in eight. This relative improvement is certified on independent data. On LLaDA,
bisection within blocks of 16 positions also raises GSM8K accuracy at
temperature 1 when all rules use 80 rounds.

The common thread is that the loss of a round comes from dependence that
remains given what has been revealed. The lower bounds count this
dependence along every trajectory a rule can take. Good schedules avoid it
by first filling the positions that cut the remaining dependence.

For practice, the results suggest three things. Positions filled together
should be far apart in the dependence structure. Positions that remain
strongly dependent given the revealed context should not share a round.
Anchors that cut the sequence first and fill the gaps in parallel follow
this rule. Second, in our results confidence is a useful order but a poor criterion
for grouping. On mixing chains, filling the most confident positions
together with leftmost tie-breaking needs linearly many rounds. On pretrained models, it raises NLL relative to
random selection with the same batch sizes. Third, the accuracy a schedule
needs on its own contexts tightens with sequence length, and random-mask
training alone does not guarantee it (Proposition~\ref{prop:avgloss}).
In our experiment, training on the contexts a schedule visits helps that
schedule, at some cost to other schedules.

The theory covers trained models only partially. Its guarantees hold
when the prediction error on the visited contexts is small, and for
pretrained models only relative improvements are certified. Determining which training objectives guarantee
this accuracy, and whether a binary alphabet admits linear depth at a fixed
positive error budget (Open Problem~\ref{op:adaptive}), remain open.

\acks{This work was supported in part by the National Natural Science
Foundation of China under Grant 12571569 and Grant U21A20455, in part by
the Natural Science Foundation of Guangdong Province of China under Grant
2024A1515011913, and in part by the Shenzhen Key Laboratory of Advanced
Machine Learning and Applications under Grant SYSPG20241211173920042.
The authors declare no
competing interests. Corresponding author: Jian Lu.}

\appendix
\makeatletter
\renewcommand\theHsection{\Alph{section}}
\makeatother

\section{Proofs for Markov Distributions and the Bernoulli Walk}\label{app:markov}

This appendix proves the results on Markov distributions and the Bernoulli
walk in Section~\ref{sec:local}.

\subsection{Bounded-Order Markov Distributions}\label{app:markovproofs}

\begin{lemma}[Block separation]\label{lem:sep}
Let $p$ be order-$m$ Markov, let $C\subseteq[n]$ with $p_C(x_C)>0$, and let
$S\subseteq[n]\setminus C$ be such that between any two coordinates of $S$
lies an $m$-block $\{s+1,\dots,s+m\}\subseteq C$. Then
\[
p(x_S\mid x_C)=\prod_{i\in S}p(x_i\mid x_C)\qquad\text{for all }x_S\in V^S.
\]
\end{lemma}

\begin{proof}[Proof of Lemma~\ref{lem:sep}]
Conditioning on an $m$-block $B=\{s+1,\ldots,s+m\}$ separates the
factorization of $p$ into $F(x_{\le s},x_B)G(x_B,x_{\ge s+m+1})$:
no potential on $m+1$ consecutive coordinates meets both sides of $B$.
Further conditioning on $x_C$ preserves this product structure.
Induct on $|S|$: the block between the smallest coordinate of $S$ and the
next one separates the former from the rest of $S$.
\end{proof}

\begin{proof}[Proof of Proposition~\ref{thm:markov}]
At each stage, choose the central $m$-block of every unrevealed interval
(the whole interval if shorter), and reveal one coordinate per chosen
block per round. Each stage uses at most $m$ rounds. Previously revealed
$m$-blocks separate the intervals, so Lemma~\ref{lem:sep} makes every
round conditionally independent, hence of zero total correlation by
Lemma~\ref{lem:tc-basic}(i); Proposition~\ref{thm:identity} gives
zero error. An interval of length $\ell>m$ leaves two intervals of length
at most $\lceil(\ell-m)/2\rceil\le\ell/2$. After
$\lceil\log_2(n/m)\rceil$ stages their lengths are at most
$n2^{-\lceil\log_2(n/m)\rceil}\le m$, and one
last stage clears them. The schedule is independent of the revealed values.
\end{proof}

\subsection{The Bernoulli Walk}\label{app:walk}

\begin{proof}[Proof of Proposition~\ref{thm:ordergap}]
(i) Proposition~\ref{thm:markov} with $m=1$ gives zero cost. Revealing one
midpoint per segment per round gives the recurrence
$L(\ell)=1+L(\lceil(\ell-1)/2\rceil)$, $L(0)=0$, whose solution is
$L(\ell)=\lceil\log_2(\ell+1)\rceil$ by induction.

(ii) Round $t$ reveals $S_t=\{a+1,\ldots,a+B_t\}$ given $x_{\le a}$.
Conditionally on $x_a$, $X_{S_t}$ is an invertible function of $B_t$ fresh
fair increments and $X_{a+j}-x_a\sim\Bin(j,\frac12)$, so
$\TC(X_{S_t}\mid x_{\le a})=\sum_{j=1}^{B_t}H(\Bin(j,\frac12))-B_t$, and
Proposition~\ref{thm:identity} gives the stated exact cost. Shannon entropy
dominates collision entropy, so Vandermonde's identity gives
$H(\Bin(j,\frac12))\ge-\log_2(4^{-j}\binom{2j}{j})\ge\frac12\log_2(\uppi j)$;
the last step holds because $a_j:=\binom{2j}{j}\sqrt{\uppi j}\,4^{-j}$
satisfies $(a_{j+1}/a_j)^2=1+1/[4j(j+1)]>1$ and $a_j\to1$ by Stirling's
formula, so $a_j\le1$. With $\log_2 B!\ge B\log_2(B/e)$, the convex
$f(B):=\frac B2\log_2B-c_0B$, $\sum_tB_t=n$ and Jensen's inequality,
\[
\KL{p_{\mathrm{walk}}}{q_\pi}
\ge\sum_{t=1}^{R}\Bigl[\tfrac12\log_2\bigl(\uppi^{B_t}B_t!\bigr)-B_t\Bigr]
\ge\sum_{t=1}^{R}f(B_t)\ge Rf(n/R)=\tfrac n2\log_2\tfrac nR-c_0n,
\]
so $\KL{p_{\mathrm{walk}}}{q_\pi}\le\eps$ implies $R\ge n2^{-2(\eps/n+c_0)}$.
\end{proof}

\begin{lemma}[Bridge structure]\label{lem:bridge}
Let $(C,x_C)$ be a context with $p_{\mathrm{walk},C}(x_C)>0$, and call the
maximal intervals of $[n]\setminus C$ \emph{segments}. With $x_0=0$ as a
permanent left anchor, a segment $T$ is either \emph{closed},
$T=\{a{+}1,\dots,b{-}1\}$ with $a\in C\cup\{0\}$, $b\in C$ and $k:=x_b-x_a$, or
\emph{open}, $T=\{a{+}1,\dots,n\}$ with $a\in C\cup\{0\}$. Conditionally on
$x_C$: (i) distinct segments are independent; the increments
$Z_{a+1},\dots,Z_b$ of a closed $T$ are uniform over arrangements with sum
$k$, and those of an open $T$ are i.i.d.\ fair bits; (ii) if $T$ is open,
or closed with $0<k<b-a$, then $H(X_i\mid x_C)>0$ for all $i\in T$ and
$I(X_i;X_j\mid x_C)>0$ for all distinct $i,j\in T$; (iii) if $T$ is closed
with $k\in\{0,b-a\}$ (\emph{degenerate}), then $H(X_i\mid x_C)=0$ for all
$i\in T$, and $T$ is contained in a maximal constant run of increments of
every consistent $x$.
\end{lemma}

\begin{proof}[Proof of Lemma~\ref{lem:bridge}]
Conditioning on the anchors fixes sums of independent increments on
disjoint blocks, proving (i). For a nondegenerate closed segment put
$m=b-a$, so $m\ge2$ since $0<k<m$, and $\sigma^2=(k/m)(1-k/m)>0$.
Exchangeability gives
$\operatorname{Var}(Z_r\mid x_C)=\sigma^2$ and
$\operatorname{Cov}(Z_r,Z_s\mid x_C)=-\sigma^2/(m-1)$ for $r\ne s$, hence
\[
\operatorname{Cov}(X_i,X_j\mid x_C)
=\sigma^2\frac{(i-a)(b-j)}{m-1}>0,\qquad a<i\le j<b.
\]
For an open segment it is $(\min(i,j)-a)/4>0$. Positive variance and
covariance imply (ii). If $k\in\{0,m\}$, all increments in the block
equal $k/m$, proving (iii).
\end{proof}

For a segment $T$ with left anchor $a$, let $s(T):=|T|+1$ be its
\emph{scale}. It equals $b-a$ when $T$ is closed by $b$. Let
$M(T):=\{a+j:\lceil s(T)/4\rceil\le j\le\lfloor3s(T)/4\rfloor\}$ be its
\emph{middle half}, so that $|M(T)|\ge|T|/2-2$.

\begin{proof}[Proof of Theorem~\ref{thm:walkconverse}(i)]
Fix $x$ in the support and track one segment: $T_1=[n]$; if
$S_t\cap T_t\ne\emptyset$, let $T_{t+1}$ be a largest unrevealed
subinterval of $T_t\setminus S_t$, with fixed tie-breaking and
$T_{t+1}=\emptyset$ if none remains, and otherwise let $T_{t+1}=T_t$. Each
$T_t$ is a segment of the current context. Put $s_t=|T_t|+1$ and let $\tau$
be the first $t$ with $T_t$ empty or degenerate; then $\tau\le R(x)+1$,
since every coordinate is revealed after round $R(x)$, and
$s_\tau\le1+\mathrm{maxrun}(x)$ by Lemma~\ref{lem:bridge}(iii). For
$t<\tau$, Proposition~\ref{thm:identity} gives $\TC(X_{S_t}\mid x_{C_t})=0$, so
Lemmas~\ref{lem:tc-basic}(ii) and~\ref{lem:bridge}(ii) give
$|S_t\cap T_t|\le1$, whence $s_{t+1}\ge s_t/2$. Therefore
$(n+1)2^{-R(x)}\le s_12^{-(\tau-1)}\le s_\tau\le1+\mathrm{maxrun}(x)$.
The union bound $\Pr[\mathrm{maxrun}\ge r]\le n2^{1-r}$ gives
$\E[\mathrm{maxrun}]\le\log_2 n+O(1)$, so taking logarithms and
expectations and using Jensen's inequality proves both bounds.
\end{proof}

\begin{lemma}[Uniform bulk-pair information]\label{lem:bulkpair}
There are absolute constants $s_0$ and $c_1>0$ with the following
property. Let $T$ be a segment of a context $(C,x_C)$ with
$p_{\mathrm{walk},C}(x_C)>0$ and $s=s(T)\ge s_0$, either open or closed with
$0<k<s$. Then $I(X_i;X_j\mid x_C)\ge c_1$ for all $i<j$ in $M(T)$.
\end{lemma}

\begin{proof}[Proof of Lemma~\ref{lem:bulkpair}]
All laws below are conditional on $x_C$. Let $a$ be the left anchor of
$T$, put $u=i-a$ and $v=j-a$, write $X_w$ for $X_{a+w}-x_a$ and
$\sigma_w^2$ for its variance, and let
$r_{uv}=\operatorname{Cov}(X_u,X_v)/(\sigma_u\sigma_v)$.
Lemma~\ref{lem:bridge} gives a bridge law for closed $T$ and a free walk
for open $T$. Flipping all increments maps $k$ to $s-k$ and preserves
$I(X_u;X_v)$, so assume $k\le s/2$ for closed $T$. For
$\lceil s/4\rceil\le u<v\le\lfloor3s/4\rfloor$, the covariance formulas in
the proof of Lemma~\ref{lem:bridge} give, for closed $T$,
$r_{uv}^2=u(s-v)/(v(s-u))\ge1/9$,
$\operatorname{Cov}(X_u,X_v)\ge k/32$ and
$\sigma_u^2\wedge\sigma_v^2\ge3k/32$, and for open $T$, $r_{uv}^2=u/v\ge1/3$
and $\sigma_u^2\wedge\sigma_v^2=u/4\ge s/16$. As the covariance is positive,
$r_{uv}\ge1/3$ in both cases.

For threshold indicators $U=\mathbf1_{\{X_u\ge h\}}$ and
$V=\mathbf1_{\{X_v\ge l\}}$, data processing and Pinsker's inequality
\citep{covertho2006} give $I(X_i;X_j\mid x_C)=I(X_u;X_v)\ge I(U;V)
\ge\frac8{\ln2}\operatorname{Cov}(U,V)^2$.
If $\bar X_w$ is $X_w$ clipped to an integer interval $[\alpha,\beta]$, then
$\bar X_w=\alpha+\sum_{h=\alpha+1}^{\beta}\mathbf1_{\{X_w\ge h\}}$, so
$\operatorname{Cov}(\bar X_u,\bar X_v)$ is a sum of threshold covariances
$\operatorname{Cov}(U,V)$, one for each pair of thresholds, and the largest
of them is at least their average.

\textbf{Bridges with $k\le10$.} Without clipping, $X_u,X_v\in\{0,\dots,k\}$
have $k$ thresholds each, so some threshold pair has
$\operatorname{Cov}(U,V)\ge(k/32)/k^2\ge1/320$.

\textbf{Bridges with $k\ge11$ and open segments.} Now $\sigma_u,\sigma_v\ge1$,
because $3k/32>1$ and $s/16\ge2$. Clip each $X_w$, $w\in\{u,v\}$, to
$[\lfloor\E X_w-100\sigma_w\rfloor,\lceil\E X_w+100\sigma_w\rceil]$, which
has at most $202\sigma_w$ thresholds. Since
$|X_w-\bar X_w|\le|X_w-\E X_w|\,\mathbf1_{\{|X_w-\E X_w|>100\sigma_w\}}$,
we have $\E(X_w-\bar X_w)^2\le\kappa_w\sigma_w^2/100^2$ with
$\kappa_w:=\E(X_w-\E X_w)^4/\sigma_w^4$. Clipping does not increase
variance, so the Cauchy--Schwarz inequality gives
$\operatorname{Cov}(\bar X_u,\bar X_v)\ge(r_{uv}-(\sqrt{\kappa_u}
+\sqrt{\kappa_v})/100)\,\sigma_u\sigma_v$. For open segments
$\kappa_w\le3$. For bridges $X_w$ is hypergeometric, so
its fourth central moment is at most that of $\Bin(w,k/s)$
\citep[Theorem~4]{hoeffding1963}, which is at most $3\tau^2+\tau$ for
$\tau=w\frac ks(1-\frac ks)$; as $\sigma_w^2=\tau\frac{s-w}{s-1}\ge\tau/4$
and $\sigma_w\ge1$, $\kappa_w\le(3\tau^2+\tau)/\sigma_w^4\le48+4/\sigma_w^2
\le52$. Hence
$\operatorname{Cov}(\bar X_u,\bar X_v)\ge(\frac13-\frac{2\sqrt{52}}{100})
\sigma_u\sigma_v\ge\sigma_u\sigma_v/6$, and some threshold pair has
$\operatorname{Cov}(U,V)\ge1/(6\cdot202^2)$. Taking $s_0=32$ and
$c_1=\frac8{\ln2}(6\cdot202^2)^{-2}$ completes the proof.
\end{proof}

\begin{proof}[Proof of Theorem~\ref{thm:walkfixed}(ii)]
Track $T_t$ as in the proof of part~(i), with
$\ell_t=|T_t|$ and $s_t=\ell_t+1$, but stop also when $s_t<s_0$. Put
$A_t:=S_t\cap T_t$, $Y_t:=\TC(X_{S_t}\mid x_{C_t})$ and
$\delta_t:=\log_2(s_t/s_{t+1})$, and call round $t$ \emph{active} if
$A_t\ne\emptyset$. By Lemma~\ref{lem:bridge}(i), reveals outside $T_t$
preserve its conditional bridge or free-walk law. The drops $\delta_t$
vanish on inactive rounds, and the chain stops with
$s_{\mathrm{end}}\le L:=\max\{s_0,\mathrm{maxrun}(x)+1\}$, so
$\sum_{\text{active }t}\delta_t=\log_2((n+1)/s_{\mathrm{end}})
\ge\log_2((n+1)/L)$.

Fix an active round; then $s_t\ge s_0$ and $T_t$ is nondegenerate, so
Lemma~\ref{lem:bulkpair} applies to $T_t$. Let $r=|A_t\cap M(T_t)|$ and,
if $r\ge2$, order these middle reveals as $i_1<\cdots<i_r$.
Lemma~\ref{lem:tc-basic}(ii), the chain rule, data processing and
Lemma~\ref{lem:bulkpair} give
\[
Y_t\ge\TC\bigl(X_{\{i_1,\dots,i_r\}}\mid x_{C_t}\bigr)
=\sum_{q\ge2}I\bigl(X_{i_q};X_{i_{<q}}\mid x_{C_t}\bigr)
\ge\sum_{q\ge2}I\bigl(X_{i_q};X_{i_{q-1}}\mid x_{C_t}\bigr)
\ge c_1(r-1)_+ ,
\]
and the last bound also holds for $r\le1$ by nonnegativity. If
$r\le\ell_t/2-2$, one of the at most $r+1$ gaps between these reveals has
length at least $(\ell_t/2-2-r)/(r+1)$ and lies in one child segment, so
$s_{t+1}\ge(\ell_t/2-1)/(r+1)$ and, as $\ell_t\ge5$,
\[
\delta_t\le\log_2\frac{s_t(r+1)}{\ell_t/2-1}\le2+\log_2(r+1)
\le2+\log_2\Bigl(2+\frac{Y_t}{c_1}\Bigr)\le3+\frac{Y_t}{2c_1\ln2}.
\]
If $r>\ell_t/2-2$, then $Y_t\ge c_1(\ell_t/2-3)$, and $s_t\ge s_0\ge32$
gives $\delta_t\le\log_2s_t\le(\ell_t/2-3)/(2\ln2)\le Y_t/(2c_1\ln2)$.
Summing over the at most $R(x)$ active rounds and using $Y_t\ge0$ gives
$\log_2\frac{n+1}{L}\le3R(x)+\frac1{2c_1\ln2}\sum_{t=1}^{R(x)}Y_t$. Taking
expectations, with $\E_p\sum_tY_t\le\eps$ by Proposition~\ref{thm:identity}
and, as in the proof of part~(i),
$\E_p[\log_2L]\le\log_2s_0+\E_p[\log_2(\mathrm{maxrun}+1)]
\le\log_2\log_2n+O(1)$, gives
$3\E_p[R]\ge\log_2(n+1)-\log_2\log_2n-O(1)-\eps/(2c_1\ln2)$. The matching
upper bound is Proposition~\ref{thm:ordergap}(i).
\end{proof}

\section{Proofs for Mixing Binary Chains}\label{app:mixing}

Coordinates take values $\pm1$; by the Markov property, conditional laws
given $x_C$ depend only on the nearest revealed coordinates.

\begin{lemma}[Pair information]\label{lem:chainpair}
Let $(C,x_C)$ be a context and let $u<v$ be coordinates outside $C$ with
$\{u+1,\dots,v-1\}\cap C=\emptyset$. Then
$I(X_u;X_v\mid x_C)\ge c_\rho\rho^{2(v-u)}$ with
$c_\rho:=(1-\rho^2)^4/(32\ln2)$.
\end{lemma}

\begin{proof}
Let $l<u$ and $r>v$ be the nearest revealed coordinates, and put
$a=\rho^{u-l}x_l$, $b=\rho^{r-v}x_r$, and $\varrho=\rho^{v-u}$.
Set $a=0$ or $b=0$ when the corresponding revealed boundary is absent.
Summing out the coordinates strictly between $u$ and $v$ gives
\[
p_\rho(X_u=s,X_v=t\mid x_C)=\frac{(1+as)(1+\varrho st)(1+bt)}{4(1+ab\varrho)},
\]
so $\operatorname{Cov}(X_u,X_v\mid x_C)=\varrho(1-a^2)(1-b^2)/(1+ab\varrho)^2
\ge\varrho(1-\rho^2)^2/4$, as $|a|,|b|\le\rho$ and $0<\varrho\le1$. For
$\pm1$ variables the joint law minus the product of the marginals equals
$st\operatorname{Cov}/4$, so their total variation distance is half the
absolute covariance, and Pinsker's inequality gives
$I(X_u;X_v\mid x_C)\ge\operatorname{Cov}^2/(2\ln2)
\ge\varrho^2(1-\rho^2)^4/(32\ln2)=c_\rho\rho^{2(v-u)}$.
\end{proof}

\begin{lemma}[Same-round pairs]\label{lem:coloring}
On a trajectory, let $I$ be an interval of $2^r$ coordinates with
$|\{t:S_t\cap I\ne\emptyset\}|\le r$. Then some round $t$ and some $i<j$ in
$S_t\cap I$ satisfy $\{i+1,\dots,j-1\}\cap C_t=\emptyset$.
\end{lemma}

\begin{proof}
An interval revealed in at most $r$ rounds without such a pair has at most
$2^r-1$ coordinates, by induction on $r$; the case $r=0$ is empty. The
earliest of these rounds reveals exactly one coordinate of the interval,
since two of them that are adjacent among those revealed in that round would
form such a pair. Removing it leaves at most two intervals, each revealed in
at most $r-1$ later rounds without such a pair, so the interval has at most
$1+2(2^{r-1}-1)=2^r-1$ coordinates.
\end{proof}

\begin{proof}[Proof of Theorem~\ref{thm:mixing}]
\textbf{Upper bound.} Use the schedule $\pi_\ell$ of
Section~\ref{sec:anchor-risk}. By~\eqref{eq:anchor-risk} with $\hat p=p$,
its error is $\Gamma_\ell\le n\rho^{2\ell}$, and it uses at most
$1+\lceil\log_2\ell\rceil$ rounds on every trajectory. Let
$\ell_0=\max\{1,\lceil\ln(n/\eps)/(2|\ln\rho|)\rceil\}$. If $\ell_0<n$, the
choice $\ell=\ell_0$ gives $n\rho^{2\ell}\le\eps$ and
$1+\lceil\log_2\ell_0\rceil=\log_2\log_2n+O_{\rho,\eps}(1)$ rounds;
otherwise $n\le\ell_0$, so $n=O_{\rho,\eps}(1)$, and revealing one
coordinate per round is exact with $n=O_{\rho,\eps}(1)$ rounds.

\textbf{Lower bound.} Fix a deterministic value-adaptive rule with
$\KL{p_\rho}{q_\pi}\le\eps$ and put
$W(x):=\sum_t\TC(X_{S_t}\mid x_{C_t})$, so that $\E W\le\eps$ by
Proposition~\ref{thm:identity}. If $R(x)\le r$, split $[n]$ into
$\lfloor n/2^r\rfloor$ disjoint intervals of length $2^r$. By
Lemma~\ref{lem:coloring}, each interval contains a pair $i<j$ with
$j-i<2^r$, revealed in one round $t$ with no coordinate between them
revealed before $t$, so Lemmas~\ref{lem:tc-basic}(ii)
and~\ref{lem:chainpair} give
$W(x)\ge A_r:=c_\rho\lfloor n/2^r\rfloor\rho^{2^{r+1}}$ whenever
$R(x)\le r$. By Markov's inequality, $p_\rho(R\le r)\le\eps/A_r$, so
$\E_p[R]\ge(r+1)(1-\eps/A_r)$. For
$r=\lfloor\log_2(\ln n/(4|\ln\rho|))\rfloor$ we have
$\rho^{2^{r+1}}\ge n^{-1/2}$, so $A_r\ge c_\rho\sqrt n/2^{r+1}\to\infty$ and
$\E_p[R]\ge\log_2\log_2n-O_{\rho,\eps}(1)$. Since $\overline D_\eps\ge D_\eps$, both
lower bounds follow.

\textbf{Zero error.} At $\eps=0$, full support forces $W=0$ on every
trajectory, so by Lemmas~\ref{lem:tc-basic}(ii) and~\ref{lem:chainpair} no
pair as in Lemma~\ref{lem:coloring} exists in $[n]$, and the induction in
the proof of Lemma~\ref{lem:coloring} gives $n\le2^{R(x)}-1$. Midpoint
bisection attains $\lceil\log_2(n+1)\rceil$ rounds.
\end{proof}

\begin{proof}[Proof of Theorem~\ref{thm:confchain}]
Put $\vartheta_d=h_2((1+\rho^d)/2)$, the conditional entropy of $X_{j+d}$
given $X_j$. It strictly increases with $d$ and
is below $1$, so $\gamma_\rho=\vartheta_2-\vartheta_1>0$.
Initially all marginals are uniform, and the tie rule selects a prefix.
After revealing $\{1,\ldots,k\}$, the confidence at $i>k$ is
$(1+\rho^{i-k})/2$, strictly decreasing in $i$. Induction therefore gives
the prefix property, even when batch sizes depend on the revealed values.

For the next $b$ positions after a nonempty prefix, the marginal entropies
are $\vartheta_1,\ldots,\vartheta_b$ and the joint entropy is
$b\vartheta_1$, hence
\[
\TC(X_S\mid x_C)
=(b-1)\gamma_\rho+\sum_{d=3}^{b}(\vartheta_d-\vartheta_2)
\ge(b-1)\gamma_\rho,
\]
with equality exactly when $b\le2$. The first batch instead costs
$\TC(X_S)=b-[1+(b-1)\vartheta_1]=(b-1)(1-\vartheta_1)\ge(b-1)\gamma_\rho$,
with equality exactly when
$b=1$. Summing $b-1$ gives $n-R(x)$ on each trajectory, and
Proposition~\ref{thm:identity} gives the result. Full support makes
equality equivalent to $|S_1|=1$ and $|S_t|\le2$ for every $t\ge2$
on every trajectory.
\end{proof}

\paragraph{Six-round comparison.} The comparison in
Section~\ref{sec:mixing} uses $n=1024$ and $\rho=0.9$. The random and
confidence rules use batch sizes $171,171,171,171,171,169$. Their costs of
$93.8$ and $159.4$ bits average $400$ target samples, with standard errors
$0.1$ and $0.3$ bits. Random positions use one fixed random permutation,
and confidence ties use a fixed random order. Round costs are computed
exactly. The confidence batches contain about $549$ adjacent pairs per
sample, versus $170$ for random batches.

\paragraph{Random tie-breaking.} This six-round confidence rule breaks
ties by a fixed random order. Its error of $159.4$ bits is below
$\gamma_{0.9}(1024-6)\approx169.5$, so it lies outside
Theorem~\ref{thm:confchain}. Leftmost ties with the same batch sizes give
$713.0$ bits.

\section{Proofs for Uniform Random Permutations}\label{app:perm}

This appendix first states the hard-cap characterization and the
variational quantities used in Theorem~\ref{thm:frontier}, and then proves
the results on permutations in Section~\ref{sec:global}.

By~\eqref{eq:permflat}, the hard-cap round--error tradeoff reduces to
integer compositions. For a composition
$(B_1,\dots,B_R)$ of $n$ into $R$ positive parts, write
$M_t=\sum_{s\ge t}B_s$ and define
\[
\cE_{n,R}:=
\min_{\substack{B_1+\cdots+B_R=n\\ B_t\ge1}}
\sum_{t=1}^{R}g(M_t,B_t),
\qquad 1\le R\le n .
\]

\begin{theorem}[Exact hard-cap characterization]\label{thm:hardcap}
For every $1\le R\le n$,
\[
\min_{\substack{\pi\text{ deterministic value-adaptive}\\
                  \max_xR(x)\le R}}
\KL{p_{\mathrm{perm}}}{q_\pi}
=\cE_{n,R},
\]
and the minimum is attained by a fixed schedule. Consequently, for every
$\eps\ge0$,
$\overline D_\eps(p_{\mathrm{perm}})=\min\{R:\cE_{n,R}\le\eps\}$.
\end{theorem}

\begin{remark}[Optimal compositions under a hard cap]\label{rem:permprofiles}
For $1\le R\le n$, some minimizer of $\cE_{n,R}$ satisfies
$B_1\ge\cdots\ge B_R$, and $\cE_{n,R}$ satisfies a dynamic-programming
recursion. If $4(n-R)\le n$, then the \emph{pair-first} composition, with
$n-R$ batches of size two followed by singletons, attains $\cE_{n,R}$.
Appendix~\ref{app:permprofiles} states and proves these results
(Proposition~\ref{prop:dporder} and Proposition~\ref{thm:phase1}).
\end{remark}

The joint-scaling limit admits a variational characterization. For
$\Lambda>0$ and $0<r\le1$, define
\begin{align*}
h_\Lambda(m)&:=\min_{\substack{B\in\mathbb N\\B\ge1}}
 \Bigl[\frac{(B-1)\log_2e}{2m}+\frac\Lambda B\Bigr],
 \quad 0<m\le1,\\
\eps_*(r)&:=\sup_{\Lambda>0}
 \Bigl[\int_0^1h_\Lambda(m)\,dm-\Lambda r\Bigr],
\end{align*}
and, for $\eps\ge0$,
\begin{equation}\label{eq:rstar}
r_*(\eps):=\inf\{r\in(0,1]:\eps_*(r)\le\eps\}.
\end{equation}
The proof of Theorem~\ref{thm:frontier} shows that
$\lim_{n\to\infty}\cE_{n,\lceil rn\rceil}=\eps_*(r)$ for every fixed
$r\in(0,1]$.
Lemma~\ref{lem:frontieranalysis} shows that the minimizing batch size is
piecewise constant. For $r\ge2/3$, $\eps_*(r)=\tfrac12\log_2(1/(2r-1))$
is the limiting cost of the pair-first composition. The condition
$4(n-R)\le n$ of Remark~\ref{rem:permprofiles} covers $R/n\ge3/4$; in the
limit, the pair-first phase extends to $r\ge2/3$ (proof of
Theorem~\ref{thm:frontier}).

\begin{proof}[Proof of Theorem~\ref{thm:perm}]
After a positive-probability history, the remaining $M$ values form a
uniform permutation. A batch of size $B$ has joint probability $1/(M)_B$
and marginal probabilities $1/M$, giving~\eqref{eq:permflat}.
Write $g(M_t,B_t)=\sum_{k<B_t}\log_2\frac{M_t}{M_t-k}$. Keeping the term
$k=B_t-1$, and using $M_t\ge R-t+1$ (as $M_2>\cdots>M_R\ge1$) and the
monotonicity of $m\mapsto m/(m+1)$, gives
\[
\sum_tg(M_t,B_t)
\ge\log_2\prod_{t=1}^R\frac{M_t}{M_{t+1}+1}
=\log_2\Bigl[n\prod_{t=2}^R\frac{M_t}{M_t+1}\Bigr]
\ge\log_2\Bigl[n\prod_{j=1}^{R-1}\frac j{j+1}\Bigr]
=\log_2\frac nR.
\]
This proves (i); Proposition~\ref{thm:identity} and Jensen give
$\eps\ge\E\log_2(n/R)\ge\log_2(n/\E R)$, that is, $\E R\ge n2^{-\eps}$.
Also, $-\ln(1-u)\ge u$ gives $g(M,B)\ge B(B-1)/(2M\ln2)$.
Using Proposition~\ref{thm:identity}, $M_t\le n$, $\sum_tB_t=n$,
Cauchy--Schwarz and Jensen,
\[
\eps\ge\frac{\E[\sum_t B_t^2]-n}{2n\ln2}
\ge\frac{n\E[1/R]-1}{2\ln2}
\ge\frac{n/\E R-1}{2\ln2},
\]
which rearranges to $\E R\ge n/(1+2\eps\ln2)$, the other bound in (ii).
\end{proof}

\begin{lemma}[Splitting a batch only helps]\label{lem:split}
For $1\le A<B\le M$,
$g(M,B)-g(M,A)-g(M{-}A,B{-}A)=(B{-}A)\log_2\frac{M}{M-A}>0$.
In particular, $\cE_{n,1}>\cE_{n,2}>\cdots>\cE_{n,n}=0$.
\end{lemma}

\begin{proof}[Proof of Lemma~\ref{lem:split}]
The identity follows from $(M)_B=(M)_A(M-A)_{B-A}$.
For $R<n$, split a nonsingleton part of a minimizing composition into
parts $1$ and $B_t-1$, giving $R+1$ parts; this strictly reduces its cost,
so $\cE_{n,R+1}<\cE_{n,R}$. Singleton batches cost zero.
\end{proof}

For partitions fixed in advance, non-strict monotonicity of the optimal
cost in the number of rounds is \citet[Theorem~2]{parallelbench2026}. For
$p_{\mathrm{perm}}$, Lemma~\ref{lem:split} adds strictness and the exact
gap.

\begin{proof}[Proof of Theorem~\ref{thm:hardcap}]
Each trajectory with $r\le R$ rounds has a composition cost at least
$\cE_{n,r}\ge\cE_{n,R}$ by~\eqref{eq:permflat} and
Lemma~\ref{lem:split}. Average using Proposition~\ref{thm:identity};
a fixed minimizing composition attains equality.
\end{proof}

Throughout, $h_\Lambda$ and $\eps_*$ are as defined at the beginning of
this appendix, $a:=\frac1{2\Lambda\ln2}$,
$A(\Lambda):=\int_0^1h_\Lambda(m)\,dm$, $B_\Lambda(m)$ is the smallest
positive integer attaining the minimum defining $h_\Lambda(m)$, and
$K_{\max}:=B_\Lambda(1)$.

\begin{lemma}[Structure of the variational limit]
\label{lem:frontieranalysis}
\begin{enumerate}[label=(\roman*),leftmargin=2.2em]
\item \emph{Phases.} The minimum defining $h_\Lambda(m)$ is attained at
$B=K$ exactly for $aK(K-1)\le m\le aK(K+1)$, and $B_\Lambda$ is a step
function with finitely many steps on $(0,1]$. The function $h_\Lambda$ is
continuous and nonincreasing in $m$, with
$0\le h_\Lambda(m)\le\min\{\Lambda,\,2\sqrt{\Lambda/m}\}$.
\item \emph{Rate.} The function
\[
r(\Lambda):=\int_0^1\frac{dm}{B_\Lambda(m)}=a(K_{\max}-1)+\frac1{K_{\max}}
\]
is continuous and nonincreasing on $(0,\infty)$, with $r(0^+)=1$ and
$r(\infty)=0$. Hence every $r\in(0,1)$ equals $r(\Lambda_r)$ for some
$\Lambda_r>0$.
\item \emph{Shape.} $\eps_*$ is convex, nonincreasing, finite and
continuous on $(0,1]$, with $\eps_*(1)=0$, and it is strictly decreasing
at every $r$ where $\eps_*(r)>0$.
\end{enumerate}
\end{lemma}

\begin{proof}[Proof of Lemma~\ref{lem:frontieranalysis}]
(i) Put $c=\log_2e/(2m)$ and $f_B=(B-1)c+\Lambda/B$. The differences
$f_{B+1}-f_B=c-\Lambda/[B(B+1)]$ increase with $B$, so $K$ minimizes $f_B$
exactly when $f_{K+1}-f_K\ge0$ and, if $K\ge2$, $f_K-f_{K-1}\le0$, that is,
$aK(K-1)\le m\le aK(K+1)$. These intervals meet at equal values, giving
continuity of $h_\Lambda$; its monotonicity follows from that of each
$f_B$. Choosing $B=1$ or $B=\lceil\sqrt{\Lambda/c}\rceil$ gives
$0\le h_\Lambda(m)\le\min\{\Lambda,2\sqrt{\Lambda c}\}
\le\min\{\Lambda,2\sqrt{\Lambda/m}\}$.

(ii) By~(i), $B_\Lambda=K$ on a phase of length $2aK$ for each
$K<K_{\max}$, so
$r(\Lambda)=\sum_{K<K_{\max}}2aK/K+[1-aK_{\max}(K_{\max}-1)]/K_{\max}
=a(K_{\max}-1)+1/K_{\max}$. For fixed $K_{\max}$ this is nonincreasing in
$\Lambda$, and at a phase boundary $aK(K+1)=1$ both adjacent forms equal
$2/(K+1)$; so $r$ is continuous and nonincreasing. For small $\Lambda$,
$K_{\max}=1$ and $r=1$; as $\Lambda\to\infty$, $K_{\max}\to\infty$ and
$a(K_{\max}-1)\le1/K_{\max}$, so $r(\Lambda)\le2/K_{\max}\to0$.

(iii) The same supremum defines $\eps_*(r)$ for every $r>0$. As a
supremum of nonincreasing affine functions of $r$, $\eps_*$ is convex and
nonincreasing on $(0,\infty)$. By~(i), $A(\Lambda)\le\int_0^12\sqrt{\Lambda/m}\,dm
=4\sqrt\Lambda$, so $\eps_*(r)\le\sup_{\Lambda>0}(4\sqrt\Lambda-\Lambda r)
=4/r$; a finite convex function on $(0,\infty)$ is continuous, so $\eps_*$ is
continuous on $(0,1]$, including at $r=1$. Since $0\le A(\Lambda)\le\Lambda$, $\eps_*(1)=0$. For $0<r<s\le1$
with $\eps_*(r)>0$, convexity gives
$\eps_*(s)\le(1-s)\eps_*(r)/(1-r)<\eps_*(r)$.
\end{proof}

\begin{proof}[Proof of Theorem~\ref{thm:frontier}]
The lower bound must hold uniformly in the realized number of rounds,
because an adaptive rule may use different compositions on different
trajectories.

\textbf{Uniform lower bound.}
Fix $\Lambda>0$, write $\xi:=\tfrac12\log_2e=1/(2\ln2)$, so that
$a=\xi/\Lambda$, and extend $h_\Lambda(0)=\Lambda$.
Lemma~\ref{lem:frontieranalysis}(i) makes $h_\Lambda$ continuous and
piecewise $C^1$ on $[0,1]$, hence Lipschitz with constant
$L_\Lambda=\Lambda^2\ln2/2$: it is constant on the $K=1$ phase, and on a
$K\ge2$ phase, where $m\ge aK(K-1)$,
\[
|h_\Lambda'(m)|=\frac{\xi(K-1)}{m^2}
\le\frac{\xi}{a^2K^2(K-1)}\le\frac{\xi}{4a^2}=L_\Lambda.
\]
For a composition put $\Sigma_g=\sum_tg(M_t,B_t)$ and $m_t=M_t/n$.
The inequality $g(M,B)\ge\xi B(B-1)/M$ from the proof of
Theorem~\ref{thm:perm}, together with $M_t\le n$, gives
\[
\Sigma_g+\frac{\Lambda R}{n}\ge\sum_t\frac{B_t}{n}h_\Lambda(m_t),
\qquad \sum_tB_t^2\le n(1+\Sigma_g/\xi).
\]
Since the intervals $(m_{t+1},m_t]$, of lengths $B_t/n$, partition
$(0,1]$, integrating the Lipschitz bound gives
\[
A(\Lambda)-\sum_t\frac{B_t}{n}h_\Lambda(m_t)
\le\frac{L_\Lambda}{2n^2}\sum_tB_t^2
\le\frac{L_\Lambda}{2n}(1+\Sigma_g/\xi).
\]
Thus, writing $u:=L_\Lambda/(2\xi n)$ and $y:=A(\Lambda)-L_\Lambda/(2n)$, and
using $\Sigma_g\le\Sigma_g+\Lambda R/n$, then $y/(1+u)\ge y-u\max\{y,0\}$ and
$0\le A(\Lambda)\le\Lambda$,
\[
\Sigma_g+\frac{\Lambda R}{n}
\ge\frac{A(\Lambda)-L_\Lambda/(2n)}{1+L_\Lambda/(2\xi n)}
\ge A(\Lambda)-\frac{C_\Lambda}{n},
\qquad C_\Lambda:=\frac{L_\Lambda}{2}(1+\Lambda/\xi).
\]
This holds for every composition and hence for
$\cE_{n,R}+\Lambda R/n$, uniformly in $R\in[n]$. Therefore
$\liminf_n\cE_{n,\lceil rn\rceil}\ge
\sup_{\Lambda>0}\{A(\Lambda)-\Lambda r\}=\eps_*(r)$ for $0<r\le1$.

\textbf{Matching construction.}
Fix $\Lambda>0$ and use batches of size $K$ while
$naK(K-1)<M_t\le naK(K+1)$, for $K=K_{\max},\ldots,2$, then singletons.
Each phase below $K_{\max}$ has length $2aKn+O(1)$, so
$R_n=n\,r(\Lambda)+O(K_{\max})$. For fixed $\Lambda$, batch sizes are
bounded and nonsingleton batches have $M_t>2an$; Taylor expansion gives
a total error $O(\sum_{t:B_t\ge2}B_t^3/M_t^2)=O(n^{-1})$.
The linearized cost is a Riemann sum with finitely many phase boundaries,
hence
\[
\sum_tg(M_t,B_t)=\mathcal J(\Lambda)+O(n^{-1}),\qquad
\mathcal J(\Lambda):=\int_0^1(B_\Lambda(m)-1)\frac{\log_2e}{2m}\,dm.
\]
Applying the uniform lower bound, with any $\Lambda'>0$, to these
compositions and letting $n\to\infty$ gives
$\mathcal J(\Lambda)\ge A(\Lambda')-\Lambda'r(\Lambda)$. As
$A(\Lambda)=\mathcal J(\Lambda)+\Lambda r(\Lambda)$,
$\eps_*(r(\Lambda))\ge A(\Lambda)-\Lambda r(\Lambda)=\mathcal J(\Lambda)
\ge\sup_{\Lambda'>0}[A(\Lambda')-\Lambda'r(\Lambda)]=\eps_*(r(\Lambda))$.
For
$\frac23\le r<1$, $K_{\max}(\Lambda_r)=2$, since $r(\Lambda)<2/K_{\max}\le\frac23$
once $K_{\max}\ge3$, and $a=r-\frac12$; so the composition above for
$\Lambda_r$ is pair-first and $\eps_*(r)=\int_{2a}^1\frac{\log_2e}{2m}dm=\frac12\log_2\frac1{2r-1}$; with
strict decrease, $r_*(\eps)=\frac{1+2^{-2\eps}}2$ for
$0\le\eps\le\frac12\log_23$.
For $0<r'<r<1$, the compositions for $\Lambda_{r'}$ have $r'n+O(1)$ rounds,
which fit the cap $\lceil rn\rceil$ for large $n$, so
$\limsup_n\cE_{n,\lceil rn\rceil}\le\eps_*(r')$ by Lemma~\ref{lem:split}.
Let $r'\uparrow r$ and use continuity. At $r=1$, sequential revealing
gives $\cE_{n,n}=0=\eps_*(1)$.

\textbf{Depth limits.}
Theorem~\ref{thm:perm}(i) and the limit just proved give
$\eps_*(r)\ge\log_2(1/r)$, hence $r_*(\eps)\ge2^{-\eps}$; in particular
$r_*(0)=1$, and both depths equal $n$ at $\eps=0$.
Fix $\eps>0$ and write $r_*=r_*(\eps)<1$, using continuity at $r=1$.
Averaging the uniform lower bound over the trajectories of a rule
attaining $D_\eps$ gives $\eps\ge A(\Lambda)-\Lambda D_\eps/n-C_\Lambda/n$
for every $\Lambda>0$ and every $n$. For $0<r'<r_*$, $\eps_*(r')>\eps$ provides
$\Lambda$ with $A(\Lambda)-\Lambda r'>\eps$, so
$D_\eps/n\ge[A(\Lambda)-\eps]/\Lambda-C_\Lambda/(\Lambda n)>r'$ for large $n$; hence
$\liminf_nD_\eps/n\ge r_*$.
Conversely, continuity gives $\eps_*(r_*)\le\eps$, and monotonicity and
strict decrease at positive values give $\eps_*(r_*+\delta)<\eps$ for
$0<\delta<1-r_*$. Hence $\cE_{n,\lceil(r_*+\delta)n\rceil}<\eps$ for all
large $n$, and Theorem~\ref{thm:hardcap} gives
$D_\eps\le\overline D_\eps\le\lceil(r_*+\delta)n\rceil$; let
$\delta\downarrow0$.
\end{proof}

\subsection{Optimal Compositions under a Hard Cap}\label{app:permprofiles}

\begin{proposition}[Structure and computation]\label{prop:dporder}
(i) For a fixed multiset of batch sizes, arranging them in nonincreasing
order minimizes the cost; some optimizer of $\cE_{n,R}$ has
$B_1\ge\cdots\ge B_R$. (ii) With $F_{n,R}:=\cE_{n,R}+\log_2 n!$,
\[
F_{n,1}=n\log_2 n,\qquad
F_{n,R}=\min_{R-1\le m\le n-1}
 \bigl\{(n{-}m)\log_2 n+F_{m,R-1}\bigr\},
\]
so the table $\{\cE_{n,R}\}_{R\le R_0}$ is computable in $O(R_0 n^2)$ time.
\end{proposition}

\begin{proof}[Proof of Proposition~\ref{prop:dporder}]
The identity $\prod_t(M_t)_{B_t}=n!$ gives
$\sum_tg(M_t,B_t)=\sum_tB_t\log_2M_t-\log_2n!$, the fixed-partition shuffle
cost of \citet[Section~4.1]{parallelbench2026}; conditioning on the
remaining count after the first batch proves (ii). Since
$p_{\mathrm{perm}}$ is exchangeable, a fixed schedule has the cost of
random-order selection with the same batch sizes, so optimal batch-size
problems of this kind are also treated by \citet[Lemma~6 and
Proposition~10]{lavenant2025} and \citet[Theorem~1.4]{chen2025schedules}.
Our contribution here is the reduction of value-adaptive rules to such
compositions (Theorem~\ref{thm:hardcap}).
For (i), take adjacent batches $a\ge b$, with $M$ coordinates remaining,
and put $x=a/M$, $y=b/M$. By this identity, the cost of order $(a,b)$ minus
that of $(b,a)$ is
$b\log_2(1-x)-a\log_2(1-y)=\frac{M}{\ln2}[y\ln(1-x)-x\ln(1-y)]\le0$,
because $u\mapsto-\ln(1-u)/u$ is nondecreasing on $(0,1)$, as $-\ln(1-u)$
is convex and vanishes at $0$.
\end{proof}

Write $s=n-R$ for the number of saved rounds and
$\varpi(m):=\log_2\frac{m}{m-1}$.

\begin{proposition}[Pairs first]\label{thm:phase1}
If $4s\le n$, then the \emph{pair-first} composition, with $s$ batches of
size two followed by $n-2s$ singletons, is optimal, and
\[
\cE_{n,n-s}\;=\;P(n,s)\;:=\;\sum_{j=0}^{s-1}\varpi(n-2j).
\]
\end{proposition}

\begin{proof}[Proof of Proposition~\ref{thm:phase1}]
Write $\widetilde\cE_{n,s}:=\cE_{n,n-s}$. The pair-first composition gives
$\widetilde\cE_{n,s}\le P(n,s)$; the reverse inequality follows by
induction on $n$ in the region $4s\le n$. For $s=0$ both sides vanish. For
$s\ge1$, Proposition~\ref{prop:dporder}(i) gives a nonincreasing optimizer,
whose first batch has $B\ge2$. Conditioning on that batch,
\[
\widetilde\cE_{n,s}\;=\;\min_{2\le B\le s+1}\
 \bigl[g(n,B)+\widetilde\cE_{n-B,\,s-B+1}\bigr],
\]
and each subproblem satisfies $4(s-B+1)\le n-B$, so the induction
hypothesis gives $\widetilde\cE_{n-B,s-B+1}=P(n-B,s-B+1)$. Define
$\Delta(B):=g(n,B)+P(n-B,s-B+1)-P(n,s)$. The pair-first recursion
$P(n,s)=\varpi(n)+P(n-2,s-1)$ gives $\Delta(2)=0$, so it suffices to show
$\Delta(B+1)\ge\Delta(B)$ for $2\le B\le s$. Since $\varpi$ is decreasing,
$P(n-B-1,s-B)\ge P(n-B,s-B+1)-\varpi(q)$ with $q:=n+B-2s$, and
$g(n,B+1)-g(n,B)=\log_2\frac n{n-B}$, so
\[
\ln2\,[\Delta(B+1)-\Delta(B)]\ge\ln\frac n{n-B}-\ln\Bigl(1+\frac1{q-1}\Bigr)
\ge\frac Bn-\frac1{q-1}>0,
\]
because $4s\le n$ gives $q-1\ge n/2+B-1$, hence $B(q-1)\ge n+2$.
\end{proof}

\section{Proofs for Balanced Binary Strings}\label{app:balanced}

Given $M$ unrevealed coordinates containing $k$ ones, all completions are
uniform. The conditional law of the values in a reveal set therefore
depends on its size and not on the identities of its positions. Sequentializing any adaptive reveal order therefore produces
uniform draws from an urn with $n/2$ ones and $n/2$ zeros. If $k(M)$ counts
the ones in its final $M$ draws, then
$k(M)\sim\mathrm{Hypergeometric}(n,n/2,M)$, and the rule visits a
subsequence of this urn trajectory. Bounds simultaneous in $M$ thus
apply to every visited state.

\begin{lemma}[Round cost on an exchangeable state]\label{lem:balcost}
Let $(C,x_C)$ be a history of positive probability with $M$ unrevealed
coordinates, $k$ of them ones, let $S\subseteq[n]\setminus C$ with
$|S|=B$, and put $\theta=\frac{k}{M}(1-\frac{k}{M})$. Then
\begin{enumerate}[label=(\roman*),leftmargin=2.2em]
\item $\displaystyle
\TC(X_S\mid x_C)\ \ge\ 2\theta\log_2e
 \sum_{j=1}^{B-1}\frac{j}{(M-1)(M-j)}$;
\item if $B\le\tfrac M8$, then $\TC(X_S\mid x_C)\le B^2/M^2$.
\end{enumerate}
\end{lemma}

\begin{proof}[Proof of Lemma~\ref{lem:balcost}]
Let $a_j$ count the ones in the first $j$ selected coordinates,
$q_j=(k-a_j)/(M-j)$ and $\bar q=k/M$. The next selected bit depends on
the first $j$ only through $a_j$, with conditional law Bernoulli$(q_j)$.
The chain rule and the hypergeometric variance formula therefore give,
with $d$ denoting Bernoulli KL,
\[
\TC(X_S\mid x_C)=\sum_{j=1}^{B-1}\E d(q_j\|\bar q),\qquad
\E q_j=\bar q,\qquad
\operatorname{Var}(q_j)=\frac{\theta j}{(M-1)(M-j)}.
\]
Pinsker gives (i). If $0<k<M$, the KL--$\chi^2$ bound gives
$d(q\|\bar q)\le(q-\bar q)^2/(\theta\ln2)$. For $2\le B\le M/8$,
$M\ge16$ and $M-1,M-j\ge7M/8$, so
\begin{align*}
\TC(X_S\mid x_C)
&\le\sum_{j=1}^{B-1}\frac{\operatorname{Var}(q_j)}{\theta\ln2}
=\log_2e\sum_{j=1}^{B-1}\frac{j}{(M-1)(M-j)}\\
&\le\frac{64\log_2e}{49M^2}\sum_{j=1}^{B-1}j
\le\frac{32\log_2e}{49}\frac{B^2}{M^2}\le\frac{B^2}{M^2}.
\end{align*}
The excluded cases $k\in\{0,M\}$ or $B\le1$ have zero cost.
\end{proof}

\begin{proof}[Proof of Theorem~\ref{thm:balanced}]
\textbf{Lower bound.} Set $M_0=\lceil c_{\mathrm{tail}}\log_2n\rceil$ with $c_{\mathrm{tail}}\ge2$
sufficiently large; finitely many $n$ with $M_0>n$ are absorbed into the
constants. A union of hypergeometric tail bounds \citep{hoeffding1963}
along the complete urn trajectory gives the good event
\[
\mathcal G:=\{k(M)/M\in[3/8,5/8]\text{ for }M_0\le M\le n\},
\qquad \Pr(\mathcal G^c)\le n^{-2}.
\]
For rounds with $M_t\ge M_0$, write $M_{t+1}=M_t-B_t$ and define
\[
\phi_t:=\frac{\bigl(B_t\wedge\frac{M_t}2\bigr)-1}{M_t},\qquad
\psi_t:=\Bigl(\log_2\frac{M_t/2}{M_{t+1}+1}\Bigr)_{\!+} .
\]
On $\mathcal G$, $\theta\ge15/64=:c_{\mathrm{bal}}$. In
Lemma~\ref{lem:balcost}(i), the terms $1\le j<B_t\wedge(M_t/2)$ contribute
at least $\theta\log_2e\,\phi_t^2$, and each term of the disjoint range
$\lceil M_t/2\rceil\le j<B_t$ contributes at least $\theta\log_2e/i$ with
$i=M_t-j$. Hence, with an empty sum read as zero,
\begin{equation}\label{eq:balcostlb}
Y_t:=\TC(X_{S_t}\mid x_{C_t})\ge\theta\log_2e\Bigl(\phi_t^2
+\!\!\sum_{i=M_{t+1}+1}^{\lfloor M_t/2\rfloor}\!\frac1i\Bigr)
\ge c_{\mathrm{bal}}\bigl(\phi_t^2+\psi_t\bigr).
\end{equation}
The drop $\Delta_t:=\log_2\bigl((M_t+1)/(M_{t+1}+1)\bigr)$, finite also
when $B_t=M_t$, satisfies, by ${-\log_2(1-u)\le2u}$ for $u\le1/2$ and
$\log_2(1+1/M_t)\le(\log_2e)/M_t$,
\[
\Delta_t=
\begin{cases}
-\log_2\bigl(1-\frac{B_t}{M_t+1}\bigr)\le\frac{2B_t}{M_t}
=2\phi_t+\frac2{M_t}, & B_t\le\frac{M_t}2,\\[3pt]
\log_2\frac{M_t+1}{M_t}+1+\log_2\frac{M_t/2}{M_{t+1}+1}
\le2\phi_t+\frac{2+\log_2e}{M_t}+\psi_t, & B_t>\frac{M_t}2,
\end{cases}
\]
so $\Delta_t\le2\phi_t+\psi_t+4/M_t$. Let $W:=\sum_tY_t$ be the total
cost, so $\E W\le\eps$ by Proposition~\ref{thm:identity}. The last round with
$M_t\ge M_0$ leaves $M_{t+1}+1\le M_0$, and distinctness of the $M_t$ gives
$\sum_{M_t\ge M_0}1/M_t\le1/M_0+\ln(1+R/M_0)$. Telescoping, the
Cauchy--Schwarz bound $\sum_t\phi_t\le(R\sum_t\phi_t^2)^{1/2}$,
and~\eqref{eq:balcostlb} therefore give, on $\mathcal G$,
\[
\log_2\frac n{M_0}\le\sum_{M_t\ge M_0}\Delta_t
\le2\sqrt{\frac{RW}{c_{\mathrm{bal}}}}+\frac W{c_{\mathrm{bal}}}
+4\ln(1+R/M_0)+\frac4{M_0}.
\]
By Markov's inequality, $\mathcal G\cap\{W\le4\eps\}$ has probability at
least $\frac34-n^{-2}\ge\frac12$. On it, either $R\ge\log_2^3n$, or
$4\ln(1+R/M_0)=O(\log\log n)$ and, for $C_\star$ large, the display gives
$4\sqrt{R\eps/c_{\mathrm{bal}}}\ge L_n(\eps)$. Hence $p(R\ge r_0)\ge\frac12$
for $r_0:=\frac{c_{\mathrm{bal}}}{16}\min\{\log_2^3n,L_n(\eps)^2/\eps\}$, and
$\E_p[R]\ge r_0/2$, which is the lower bound with $c=c_{\mathrm{bal}}/32$.

\textbf{Upper bound.} Use $B_t=\lceil\gamma M_t\rceil$ with
$\gamma=\eps/(16\log_2n)\le1/16$ until $M_t\le M_0$, then singletons.
Geometric decay $M_{t+1}\le(1-\gamma)M_t$ and $M_t>M_0\ge1$ throughout
the first phase bound it by $1+\ln(n)/\gamma\le\log_2(n)/\gamma$ rounds,
the last step because $\gamma\le1/16$ and $n\ge2$; hence
$R\le16\log_2^2n/\eps+(c_{\mathrm{tail}}+1)\log_2n
\le C_\star(\log_2^2n/\eps+\log_2n)$ for
$C_\star\ge\max\{16,c_{\mathrm{tail}}+1\}$.
A nonsingleton batch has $\gamma M_t>1$, so
$B_t<\gamma M_t+1<2\gamma M_t$ and $B_t/M_t\le2\gamma\le1/8$;
Lemma~\ref{lem:balcost}(ii) bounds its cost by $4\gamma^2$.
At most $\log_2(n)/\gamma$ batches are nonsingleton, so
$W\le4\gamma\log_2n=\eps/4$ on every trajectory.
Proposition~\ref{thm:identity} and the fixed schedule give both upper bounds.

Finally, let $c'=\min\{1,1/(4C_\star)\}$. For $1/\log_2n\le\eps\le c'\log_2n$ and
$n$ large, $L_n(\eps)\ge\frac12\log_2n$, so the lower bound is at least
$\frac c4\eps^{-1}\log_2^2n$ and the upper bound at most
$2C_\star\eps^{-1}\log_2^2n$.
\end{proof}

\section{Proofs for Binary One-Hot Blocks}\label{app:onehot}

Throughout this appendix, $p=p_{b,L}^{\mathrm{hot}}$ consists of $b$
independent one-hot blocks of length $L$, and $p_m^{\mathrm{hot}}=p_{m,m}^{\mathrm{hot}}$.

In an unresolved block with $M\ge2$ candidate locations, selecting $k$
candidates yields the zero vector with probability $1-k/M$ and each of
the $k$ unit vectors with probability $1/M$. Thus the conditional total
correlation is $G(M,k):=k\,h_2(1/M)-h_2(k/M)-\frac kM\log_2k$, with
$G(M,0):=0$.

\begin{lemma}[One-round one-hot cost]\label{lem:hotcost}
For every $M\ge2$ and $0\le k\le M$,
$G(M,k)\ge c_{\mathrm{hot}}\bigl((k-1)_+/M\bigr)^2$ with $c_{\mathrm{hot}}:=1/(2e\ln2)$.
\end{lemma}

\begin{proof}
For $k\ge2$, let $P$ be the selected coordinates' law, $Q$ the product of
their Bernoulli$(1/M)$ marginals, and $A$ the event of at least two ones.
Then $P(A)=0$ and
$Q(A)\ge\binom k2M^{-2}(1-\frac1M)^{k-2}\ge k(k-1)/(2eM^2)$, since $k\le M$
and $(1-1/M)^{M-2}\ge e^{-1}$. Data processing gives
$G(M,k)=\KL{P}{Q}\ge-\log_2(1-Q(A))\ge(\log_2e)Q(A)\ge c_{\mathrm{hot}}(k-1)^2/M^2$.
For $k=0,1$, both sides are zero.
\end{proof}

Fix a block $a$ and a global round $t$, and pad the trajectory with empty
rounds after termination. Let $Z_{a,t}$ indicate that the block is
unresolved immediately before round $t$. If $Z_{a,t}=1$, let $M_{a,t}\ge2$
be its number of candidate locations and $k_{a,t}$ the number of them
selected in round $t$; if $Z_{a,t}=0$, set $M_{a,t}=1$ and $k_{a,t}=0$.
Put $u_{a,t}:=(k_{a,t}-1)_+/M_{a,t}$.

\begin{lemma}[One free candidate per round]\label{lem:hotfree}
For every deterministic value-adaptive rule, every block $a$ and every
round $t$, $\E_p[Z_{a,t}/M_{a,t}]\le1/L$.
\end{lemma}

\begin{proof}
Fix the other blocks' latent locations $J_{-a}=j_{-a}$. Simulate the rule
using those values and answering zero in block $a$ until at most one
candidate remains. Let $U_t$ be its candidate set before round $t$ if
unresolved, and $\varnothing$ otherwise. For $J_a\in U_t$ the true history
agrees with the simulation until round $t$, because every earlier answer in block $a$ is zero and the rule is
deterministic, so it makes the same successive queries. Hence $Z_{a,t}=1$
and $M_{a,t}=|U_t|$. For $J_a\notin U_t$ the one has been revealed or
inferred, so $Z_{a,t}=0$.
Since $J_a$ is uniform on $[L]$ conditional on $J_{-a}$,
\[
\E_p[Z_{a,t}/M_{a,t}\mid J_{-a}=j_{-a}]
=\sum_{j\in U_t}\frac1{L|U_t|}\le\frac1L,
\]
where an empty sum is zero. Average over $J_{-a}$.
\end{proof}

\begin{proof}[Proof of Theorem~\ref{thm:onehot}]
Fix a deterministic value-adaptive rule $\pi$ with $\KL{p}{q_\pi}\le\eps$.
For block $\mathcal B_a$ put
$K_a:=\E_p\bigl[\sum_{t\ge1}\TC(X_{S_t\cap\mathcal B_a}\mid x_{C_t})\bigr]$,
where empty intersections contribute zero; the summand is
$G(M_{a,t},k_{a,t})$ while the block is unresolved and zero otherwise.
Lemma~\ref{lem:tc-basic}(iii) and Proposition~\ref{thm:identity} give
$\sum_{a=1}^bK_a=\KL{p}{q_\pi}\le\eps$, so some block $a$ has
$K_a\le\eps/b$; fix it.

Let $T_a\le R$ be the round when block $a$'s one is revealed or only one
candidate remains. With $M$ candidates and batch size $k$, the resolution
probability is $k/M$ for $k\le M-2$ and one otherwise; in either case it
is at most $(k-1)_+/M+2/M$. Summing the disjoint resolution events gives
$p(R\le t)\le p(T_a\le t)\le\sum_{s\le t}\E_p[Z_{a,s}(u_{a,s}+2/M_{a,s})]$,
and Lemma~\ref{lem:hotfree}, Cauchy--Schwarz, Lemma~\ref{lem:hotcost} and
$K_a\le\eps/b$ give, for every $t\ge0$,
\[
p(R\le t)\le\sum_{s\le t}\E_p[Z_{a,s}u_{a,s}]+\frac{2t}L
\le\sqrt{\frac{tK_a}{c_{\mathrm{hot}}}}+\frac{2t}L
\le\sqrt{\frac{t\eps}{c_{\mathrm{hot}}b}}+\frac{2t}L.
\]
For $t=\lfloor\min\{L/16,\,c_{\mathrm{hot}}b/(16(1+\eps))\}\rfloor$ the two terms are at
most $\frac14$ and $\frac18$, so $p(R>t)\ge\frac58$ and
$\E_p[R]\ge\frac58(t+1)\ge c\min\{L,b/(1+\eps)\}$ with $c=5c_{\mathrm{hot}}/128$.

For the upper bound, the fixed schedule $S_t=\{(a,t):a\in[b]\}$,
$t=1,\ldots,L$, reveals one coordinate per block in every round.
Conditional independence across blocks gives zero cost, so
$\overline D_\eps(p)\le\overline D_0(p)\le L$.
Finally, for fixed $\eps$ choose $L\asymp n^\alpha$ and $b\asymp n^{1-\alpha}$ along
integer subsequences; when $0<\alpha\le1/2$, the lower and upper bounds are
both of order $n^\alpha$.
\end{proof}

\section{Proofs for the Depth Spectrum}\label{app:generic}

\begin{proof}[Proof of Proposition~\ref{prop:generic}]
For fixed $(C,x_C,\{i,j\})$, conditional independence is the polynomial
system
\[
p_{\{i,j\}\cup C}(a,b,x_C)p_C(x_C)
=p_{\{i\}\cup C}(a,x_C)p_{\{j\}\cup C}(b,x_C),\qquad(a,b)\in V^2.
\]
At the uniform law on two strings agreeing on $C$ and differing at both
$i,j$, which exist as $|V|\ge2$, one equation reads $1/2=1/4$. Its polynomial is therefore nonzero
on the simplex's affine hull and has a Lebesgue-null zero set. The finite
union over all triples is null. Outside it, no pair is conditionally
independent at any history. For full-support laws,
Proposition~\ref{thm:identity} requires zero total correlation in every round of a zero-error rule on
every trajectory. By Lemma~\ref{lem:tc-basic}(ii), every pair in the round
is then conditionally independent, so every round is a singleton. Hence
$R(x)=n$ for all $x$ and $D_0=\overline D_0=n$.
\end{proof}

\begin{proof}[Proof of Theorem~\ref{thm:binarynear}]
Put $\nu=\log_2n$, $K_0=n/4$, $\mathcal K=\{K_0,\ldots,2K_0-1\}$, $s=4\nu$, $\eta=1/n$,
$a=1-\eta$ and $u_k=k-s\ge1$; subscripts on $\E$ and $H$ indicate the law, and
$h_2$ is the binary entropy of Section~\ref{sec:mixing}.

\textbf{Components.} For $k\in\mathcal K$ let $\mu_k$ be the empirical law of
$2^k$ independent uniform vectors in $\{0,1\}^n$, so $H(\mu_k)\le k$. For
$C\subseteq[n]$ and $z\in\{0,1\}^C$ write $[C,z]=\{x:x_C=z\}$. If
$|C|\le u_k$, the number of vectors in $[C,z]$ is binomial with mean
$2^{k-|C|}\ge n^4$, so by the Chernoff bound it leaves
$[1-\eta,1+\eta]$ times its mean with probability at most $2e^{-n^2/3}$.
Hence, except with probability at most $2K_0\,3^ne^{-n^2/3}$, which is below
$1$ and tends to $0$,
\begin{equation}\label{eq:bnreg}
a\,2^{-|C|}\le\mu_k([C,z])\le(1+\eta)\,2^{-|C|}
\qquad(|C|\le u_k,\ k\in\mathcal K).
\end{equation}
For $|C|<u_k$ and $i\notin C$, applying~\eqref{eq:bnreg} to both
extensions of $[C,z]$ gives
$\mu_k(X_i=1\mid x_C)\in[\frac{1-\eta}2,\frac{1+\eta}2]$, hence
$H_{\mu_k}(X_i\mid x_C)\ge h_2(\frac{1-\eta}2)\ge1-\eta^2/\ln2\ge a$. Let
$p_n=\frac34K_0^{-1}\sum_k\mu_k+\frac14U$ with $U$ uniform. Then
$p_n\ge2^{-n}/4$ and, with $\bar k=K_0^{-1}\sum_kk<n/2$,
\[
H(p_n)\le h_2(\tfrac34)+\tfrac34(\bar k+\log_2K_0)+\tfrac n4.
\]

\textbf{Thresholds.} Fix $\pi$, let $N_t$ be the number of coordinates
revealed in rounds $1,\ldots,t$, $T_k=\min\{t:N_t\ge u_k\}$ and
$\tau_k=N_{T_k}$. Since $\pi$ sees only revealed values, the history before
round $T_k$ is a cylinder $[C,z]$ with $|C|<u_k$ on which $T_k$ and
$\tau_k$ are constant, and these cylinders partition $\{0,1\}^n$;
so~\eqref{eq:bnreg} gives
\begin{equation}\label{eq:bnstop}
\E_{\mu_k}T_k\ge a\,\E_UT_k,\qquad\E_{\mu_k}\tau_k\ge a\,\E_U\tau_k.
\end{equation}
On a trajectory, let $m_t=|\{k:T_k=t\}|$ and
$\bar T=K_0^{-1}\sum_kT_k=K_0^{-1}\sum_ttm_t$. The thresholds crossed in round
$t$ are $m_t$ consecutive integers, so their overshoots $\tau_k-u_k$ are
$m_t$ distinct nonnegative integers and sum to at least $m_t(m_t-1)/2$. By Markov's inequality at least $K_0/2$ thresholds
are crossed in the first $2\bar T$ rounds, and the Cauchy--Schwarz
inequality gives
\begin{equation}\label{eq:bnover}
\bar e:=\frac1{K_0}\sum_k(\tau_k-u_k)\ge\frac1{2K_0}\sum_tm_t(m_t-1)
\ge\frac{K_0}{16\bar T}-\frac12.
\end{equation}

\textbf{Entropy accounting.} By Proposition~\ref{thm:identity} and the chain rule
$\E_{p_n}\sum_tH(X_{S_t}\mid x_{C_t})=H(p_n)$, we have
$\KL{p_n}{q_\pi}=\mathcal A(p_n)-H(p_n)$ with
$\mathcal A(Q):=\E_Q\sum_t\sum_{i\in S_t}H_Q(X_i\mid x_{C_t})$. At every
history $[C,z]$, $p_n(\cdot\mid x_C)$ mixes the components' conditional
laws with posterior weights, and a posterior weight times $p_n([C,z])$ is
the prior weight times the component's probability of $[C,z]$; concavity
of entropy therefore gives $\mathcal A(p_n)\ge\frac3{4K_0}\sum_k\mathcal A(\mu_k)+\frac14\mathcal A(U)$,
where $\mathcal A(U)=n$. Rounds $t\le T_k$ start with fewer than $u_k$ revealed
coordinates, so $\mathcal A(\mu_k)\ge a\,\E_{\mu_k}\tau_k\ge a^2(u_k+\E_U(\tau_k-u_k))$
by~\eqref{eq:bnstop}. Let $d=\E_{p_n}R\ge\frac3{4K_0}\sum_k\E_{\mu_k}T_k
\ge\frac34a\,\E_U\bar T$. Jensen's inequality in~\eqref{eq:bnover} gives
$\E_U\bar e\ge\frac{3aK_0}{64d}-\frac12$, and hence, using
$(1-a^2)\bar k\le1$,
\[
\KL{p_n}{q_\pi}\ge\tfrac34a^2(\bar k-s+\E_U\bar e)
-\tfrac34(\bar k+\log_2K_0)-h_2(\tfrac34)
\ge\frac{9a^3K_0}{256\,d}-s-\log_2K_0-3.
\]
For $n\ge256$, $9a^3/256\ge1/32$, so
\begin{equation}\label{eq:bnpol}
\KL{p_n}{q_\pi}\ge\frac n{128d}-5\nu-1
\end{equation}
for every realization satisfying~\eqref{eq:bnreg}, and
$\KL{p_n}{q_\pi}\le\eps$ gives the lower bound.

\textbf{Upper bound.} Reveal coordinates in the order $1,\ldots,n$; for a
revealed prefix $z\in\{0,1\}^\ell$ write $[z]=[[\ell],z]$. Let
$B=\lfloor \nu/2\rfloor$, $b_0=\lceil\log_2\nu\rceil+1$ and $\delta=n^{-3}$, and let
$W_\ell$ be the set of prefixes $z\in\{0,1\}^\ell$ of codewords of the
$\mu_k$ with $k\le\ell-b_0$. The rule reveals all remaining coordinates if
some completion of $z$ has conditional probability at least $1-\delta$ (the
terminal rule); otherwise the next coordinate if $z\in W_\ell$, and otherwise
the next $v=\min\{B,n-\ell\}$ coordinates. Single rounds cost nothing, and
the terminal round costs at most $nh_2(\delta)$. For a block round at
$z\notin W_\ell$, let $P_z$ be the conditional law of
$X_{\ell+1},\ldots,X_{\ell+v}$ and $U_v$ the uniform law on $\{0,1\}^v$. Then
$\TC(P_z)\le\KL{P_z}{U_v}\le2^v\sum_w(P_z(w)-2^{-v})^2/\ln2$; as
$p_n([z])\ge2^{-\ell}/4$ and codebooks with $k\le\ell-b_0$ put no mass on
$[z]$,
\[
p_n([z])\,\TC(P_z)\le\frac{4\cdot2^{\ell+v}}{\ln2}\sum_w\Bigl(\frac3{4K_0}
\sum_{k>\ell-b_0}\Delta_k(zw)\Bigr)^2,\qquad
\Delta_k(zw)=\mu_k([zw])-2^{-v}\mu_k([z]).
\]
Each prefix length is visited at most once, so we may sum over all $z$ and
then average over the independent random codebooks; under this average,
written $\E$, $\E\Delta_k(zw)=0$ and $\E\Delta_k(zw)^2\le2^{-\ell-v-k}$. The
expected block cost is at most
$\sum_{\ell=0}^{n-1}\frac{9\cdot2^{\ell+v}}{4K_0^2\ln2}\sum_{k>\ell-b_0}2^{-k}
\le\frac{36\cdot2^{B+b_0}}{n\ln2}\le\frac{144\,\nu}{\sqrt n\ln2}=:g_n$.
For the round count, $U(W_\ell)\le\sum_{k\le\ell-b_0}2^{k-\ell}\le1/\nu$. Under
$\mu_j$ and for $\ell<j+b_0$, $W_\ell$ does not depend on the $j$th codebook, so
$\E\mu_j(W_\ell)\le1/\nu$; from $\ell=j+b_0$ on, every prefix of $X\sim\mu_j$ lies
in $W_\ell$. A path has at most $\lceil n/B\rceil$ block rounds. The block that crosses
$j+b_0$ ends before $j+b_0+B\le\ell_j:=j+5\nu<n$. Unless the rule has
stopped, it therefore reaches $\ell_j$ after at most $5\nu$ single rounds.
At most $n$ rounds remain if the terminal rule fails there. Since
$p_n(X)\ge\frac3{4K_0}2^{-j}$ and the other codewords are independent and
uniform, $\E[p_n([X_{[\ell_j]}])-p_n(X)]\le2^{-\ell_j}$, so
$\E[1-p_n(X)/p_n([X_{[\ell_j]}])]\le\frac{4K_0}32^{j-\ell_j}$, and by Markov's
inequality the terminal rule fails at $\ell_j$ with $\mu_j$-probability at
most $1/(3n)$ in expectation. Hence
$\E\,\E_{p_n}R\le\lceil n/B\rceil+1+n/\nu+\frac{15}4\nu+\frac14\le5n/\nu$. By
Markov's inequality, for large $n$ some codebooks satisfy~\eqref{eq:bnreg},
have block cost at most $\sqrt{g_n}$, and have $\E_{p_n}R\le10n/\nu$; we fix
such codebooks (and otherwise any satisfying~\eqref{eq:bnreg}). The rule
then has error at most $\sqrt{g_n}+nh_2(\delta)\to0$ and at most
$10n/\log_2n$ rounds on average.

\textbf{Padding.} For $\alpha\in(0,1)$ let
$n_0=2^{\lfloor\log_2(n^\alpha\log_2n)\rfloor}$, so that
$n_0/\log_2n_0=\Theta_\alpha(n^\alpha)$, and for $256\le n_0\le n$ let
$p'_n=p_{n_0}\otimes U_{n-n_0}$ (any law with $p'_n\ge2^{-n}/4$ otherwise).
Revealing the $n-n_0$ fair bits in the first round of the rule above adds
no error and no round. For a rule $\pi$ for $p'_n$ and fixed values $y$ of
these bits, deleting rounds that reveal only fair bits gives a deterministic
rule $\pi_y$ for $p_{n_0}$ with $q_\pi(x,y)=2^{n_0-n}q_{\pi_y}(x)$, so
$\KL{p'_n}{q_\pi}=\E_y\KL{p_{n_0}}{q_{\pi_y}}$ and
$\E_{p'_n}R\ge\E_y\E_{p_{n_0}}R_{\pi_y}$. Applying~\eqref{eq:bnpol} to each
$\pi_y$ and Jensen's inequality to the convex decreasing map
$d\mapsto n_0/(128d)$ gives $D_\eps(p'_n)\ge n_0/(128(\eps+5\log_2n_0+1))$.
\end{proof}

\section{Experimental Details and Checks}\label{app:empchecks}
\label{sec:emprobust}

This appendix gives the data, training, and certification settings of
Section~\ref{sec:decoders-part}, the pseudo-cost measurements, and the
task-accuracy protocol.

\subsection{Data, Training, and Certification}\label{app:empsettings}

\paragraph{Data and likelihoods.}
Each non-WikiText domain has 256 source-disjoint excerpts. WikiText has
128 excerpts per larger model and 256 held-out excerpts for 0.5B. The 0.5B
model is converted to masked diffusion by the shift-preserving method of
\citet{dream2025}. The
theory applies to the conditional law given the first token. Larger-model
pretraining overlap is unknown.
The non-WikiText sources are CNN/DailyMail 3.0.0 test articles, English
OpenAssistant train/validation trees (one deterministic path per tree),
CodeSearchNet Python test functions (one eligible function per
repository), and GovReport test bodies. Deduplication precedes
source-level sampling. The 128 long-context reports are eligible at length
2048 and evaluated at all three prefix lengths. Models see the same
source texts, but their tokenizers differ. Position zero is supplied;
Dream and 0.5B use shifted logits, LLaDA same-position logits. Paired NLL
intervals are document-level mean differences plus or minus 1.96 standard
errors after averaging position seeds.

Table~\ref{tab:abs-nll} gives absolute NLL. The batch-matched
confidence-minus-random differences range from $2.12$ (Dream WikiText) to
$2.69$ (LLaDA dialogue). On WikiText, confidence top-64 minus random-64 is $1.42$ for 0.5B and
$1.61$ for LLaDA and is positive on every article. Sequentializing these
orders gives $-0.033$ and $-0.072$. For top-16 against random-16 the
values are $-0.035$ and $-0.090$. The sequential reference in
Section~\ref{sec:empdecoders} uses the first 32 excerpts per non-WikiText
domain for LLaDA and 64 for Dream and 0.5B. Anchor-minus-sequential NLL
has 95\% half-widths at most $0.06$; the 0.5B code reference is $14.54$
bits/token, versus $17.2$ for its uniform vocabulary law.

\begin{table}[t]
\centering\small
\setlength{\tabcolsep}{5pt}
\begin{tabular}{llrrrrr}
\hline
Model & Domain & $\mathsf A_5$ & $\mathsf R_5$ & $\mathsf R_8$ & confidence & bisection\\
\hline
LLaDA & Code     & $2.74$ & $3.06$ & $2.92$ & $5.39$ & $2.68$\\
      & Dialogue & $2.80$ & $3.20$ & $3.21$ & $5.89$ & $2.61$\\
      & News     & $4.33$ & $4.98$ & $4.52$ & $7.59$ & $4.21$\\
      & Reports  & $3.85$ & $4.44$ & $4.12$ & $6.99$ & $3.73$\\
      & WikiText & $4.24$ & $4.83$ & $4.34$ & $7.34$ & $4.17$\\
\hline
Dream & Code     & $3.46$ & $3.78$ & $3.55$ & $6.33$ & $3.37$\\
      & Dialogue & $2.92$ & $3.30$ & $3.17$ & $5.52$ & $2.82$\\
      & News     & $4.50$ & $5.10$ & $4.53$ & $7.30$ & $4.44$\\
      & Reports  & $3.84$ & $4.42$ & $4.01$ & $6.61$ & $3.76$\\
      & WikiText & $4.41$ & $5.01$ & $4.46$ & $7.13$ & $4.35$\\
\hline
\end{tabular}
\caption{Mean conditional NLL in bits per predicted token, lower is
better, on the excerpts of Table~\ref{tab:parallel-likelihood}.
Confidence selection uses the batch sizes of $\mathsf A_5$, and bisection
takes nine rounds. The $\mathsf R_5$ column averages three position plans
and the $\mathsf R_8$ column uses one, so differences of the $\mathsf R_8$
column from $\mathsf A_5$ can deviate from
Table~\ref{tab:parallel-likelihood} by up to $0.02$.}
\label{tab:abs-nll}
\end{table}

\paragraph{Schedules.} Spacing 16 is the largest allowed by
$1+\lceil\log_2\ell\rceil=5$. $\mathsf R_5$ randomly permutes positions
with batch sizes $31,32,64,128,256$, and $\mathsf R_8$ uses
$64,\ldots,64,63$; random results average three position seeds per
document. Confidence selection ranks positions by $\max_v\hat
p(v\mid\cdot)$ before drawing tokens. Each fixed position plan and this
confidence rule satisfy Definition~\ref{def:policy}. LLaDA's reference
remasking agrees with this selection at temperature zero; at positive
temperature it selects using current-round draws and falls outside that
definition.

\paragraph{Further likelihood results.} The smallest seed-averaged
$\mathsf A_5-\mathsf R_8$ gain is $0.032$ on Dream news (paired 95\%
interval $[0.013,0.051]$). On 128 reports evaluated at lengths 512, 1024,
and 2048, the batch-matched advantage remains $0.50$--$0.60$ bits/token in
five rounds. Relative to one-position-per-round decoding, five rounds lose
$0.65$--$0.84$ bits/token for LLaDA and $0.88$--$1.10$ for Dream on
non-WikiText domains. On code, 0.5B anchors are $1.03$ worse than $\mathsf
R_5$.

\paragraph{Training budgets and controls.}
Anchor alignment selects round $t$ with probability $B_t/511$ and averages
its $B_t$ target losses; random alignment uses one fixed $\mathsf R_5$
position plan, and mixed training alternates ordinary and anchor-aligned
objectives. Across 18 runs, thinning ordinary supervision worsens both decoders.
Anchor alignment gives interactions of $-0.374$ to $-0.380$ across seeds.
Aligning to one random position seed gives no advantage over the other two
random seeds. Longer training also worsens WikiText NLL:
$0.30$--$0.41$ for ordinary training and $0.41$--$0.51$ for anchor
alignment. Mixed training loses less ($0.17$--$0.22$), but has no
equal-step ordinary control at 1218 steps.
Ordinary 1000-step training processes 32,000 sequences and 8.54--8.62
million supervised tokens (about 267 per sequence); alignment supervises
about 172. The extended aligned and mixed arms approximately match this
token budget at 1558 and 1218 steps. Endpoint counts are approximate
because logs record periodic counts. Thinning independently retains each
ordinary loss term with probability 0.64 and divides its weight by 0.64,
preserving the expected objective but changing gradient noise. Each arm
uses its own training length for the cosine learning-rate schedule, so
optimization trajectories also differ.

The longer ordinary control was added after inspecting the first training
seed. Two seeds used different hardware from the aligned arms; the
ordinary reference differs by about $0.006$ bits/token across machines.
The two pure alignment arms share hardware, length, and supervised-token
budget within each seed. Their mean interaction difference is $-0.377$,
independent of the post-hoc control. Against that control, anchor
alignment improves $\mathsf A_5$ by $0.090$--$0.099$ and worsens
$\mathsf R_5$ by $0.262$--$0.273$ across seeds.

\paragraph{Known-target diagnostic.}
A four-layer, width-128 bidirectional transformer is trained for 40,000
steps on the $\rho=0.9$ mixing chain. Each training mask independently
reveals positions at a rate drawn uniformly from $[0,1]$. The spacing-32
schedule (six rounds) anchors at $1,33,65,\ldots$; it is the mirror image
of $\pi_{32}$ and has the same $\Gamma_{32}$. Splitting a total budget of
$0.1$ bit equally between dependence and prediction allows
$\mathcal M(\hat p,\pi)\le0.05$ bit. Estimates using 256 target samples
and true conditionals are $0.005$ and $0.041$ bits at lengths 128 and 256,
but $1.68$ and $3.44$ for two training seeds at length 512. The first two
empirical estimates fall below the threshold; neither length-512 estimate
does.

\paragraph{Certificate construction.}\label{app:certificate}
Fix $K$ candidate decoders, a reference decoder for each, smoothing
weights, and clipping levels before drawing $N\ge2$ independent,
identically distributed validation documents from a declared population.
Smooth the conditionals of candidate $j$ as $\hat p_j^\eta=(1-\eta)\hat
p_j+\eta/|V_j|$ with $0<\eta<1$ and vocabulary $V_j$. Its selection rule
stays fixed. On a document, let $\mathrm{nll}_j$ and $\mathrm{nll}_j^{\rm
raw}$ be the mean per-token losses of the smoothed and the unsmoothed
candidate. Let $\mathrm{nll}_{0,j}$ be that of the reference, with finite
expectation. Then $0\le\mathrm{nll}_j\le u_j:=\log_2(|V_j|/\eta)$, and
$\kappa_j=\E(\mathrm{nll}_j-\mathrm{nll}_{0,j})$ equals the difference of
conditional forward KL divergences per token. To reuse unsmoothed scores,
take the computable envelope
$G_j=\min\{\mathrm{nll}_j^{\rm raw}-\log_2(1-\eta),u_j\}\ge\mathrm{nll}_j$;
direct smoothed scoring instead uses $G_j=\mathrm{nll}_j$. For fixed
$c_j>0$, set $Z_j=\max\{G_j-\mathrm{nll}_{0,j},-c_j\}$. These document
scores lie in $[-c_j,u_j]$ and have mean at least $\kappa_j$. The empirical
Bernstein inequality of \citet[Theorem~4]{maurerpontil2009}, followed by a
union bound, gives
\begin{equation}\label{eq:relative-certificate}
\kappa_j\le U_j:=\bar Z_j+
\sqrt{\frac{2\hat\sigma_j^2\ln(2K/\alpha)}{N}}
+\frac{7(u_j+c_j)\ln(2K/\alpha)}{3(N-1)}
\quad\text{for all }j
\end{equation}
with probability at least $1-\alpha$, where $\bar Z_j$ and
$\hat\sigma_j^2$ are the sample mean and the sample variance with
denominator $N-1$. No independence between comparisons is required.

\paragraph{Confirmation settings.}
Before sampling, the preregistration fixed $K=2$, $\alpha=0.05$,
$\eta=10^{-6}$, clipping at $-c_j=-2$, acceptance when $U_j\le0$, and 511
predicted tokens. The candidate uses spacing 16 and five rounds; the
reference uses one fixed random plan with batches $64,\ldots,64,63$.
The population is the test SFT split of UltraChat-200k, 20,876 deduplicated
dialogues long enough for both tokenizers. Candidates, parameters, and the
sampling seed were preregistered before 4096 uniform draws with
replacement; the complete protocol was fixed after sampling and before
scoring. The candidate is $\mathsf A_5$ with each prediction mixed with
the uniform law at weight $10^{-6}$; the reference is unsmoothed $\mathsf
R_8$ with one fixed plan.

\subsection{Pseudo-Cost Measurements}\label{app:pseudocost}

\paragraph{Sequentialized comparison.}
The pseudo-cost compares a parallel decoder with the decoder that fills
each of its batches one position at a time.
At each state, fix an order $\sigma=(s_1,\ldots,s_k)$ of $S=\pi(C,x_C)$,
and write $s_{<j}=\{s_1,\ldots,s_{j-1}\}$. Its \emph{ordered pseudo-cost} is
\begin{equation}\label{eq:orderedtc}
\widehat{\tc}_{\sigma}(x_S\mid x_C)
:=\sum_{j=1}^k\Bigl[\log\hat p_{s_j}(x_{s_j}\mid x_{C\cup s_{<j}})
-\log\hat p_{s_j}(x_{s_j}\mid x_C)\Bigr].
\end{equation}
If $\hat p$ consists of the conditionals of a joint law $r$, this is its
realized total correlation and does not depend on $\sigma$. For $r=p$ it
is $\tc(x_S\mid x_C)$. Otherwise its $p$-mean can be negative.
The \emph{sequentialization} $\pi^{\mathrm{seq}}$ reveals $S$ in order
$\sigma$, calling $\pi$ again only after $S$ is exhausted. It is a
deterministic value-adaptive rule: replaying $\pi$ on the revealed
values recovers the current batch and the positions already revealed.

\begin{corollary}[Sequentialized comparison]\label{cor:seq}
Suppose every $\hat p_i(\cdot\mid x_C)$ has full support, and let
$\sigma_t=(s_1,\ldots,s_{|S_t|})$ be the order at round $t$.
For $p(x)>0$,
\begin{equation}\label{eq:seqpoint}
\log\frac{\hat q_{\pi^{\mathrm{seq}}}(x)}{\hat q_\pi(x)}
=\sum_{t=1}^{R(x)}\widehat{\tc}_{\sigma_t}\bigl(x_{S_t}\mid x_{C_t}\bigr),
\end{equation}
and
\begin{equation}\label{eq:seqdecomp}
\begin{aligned}
&\KL{p}{\hat q_\pi}-\KL{p}{\hat q_{\pi^{\mathrm{seq}}}}
=\E_{x\sim p}\sum_{t=1}^{R(x)}
\widehat{\tc}_{\sigma_t}\bigl(x_{S_t}\mid x_{C_t}\bigr)\\
&\quad=\E_{x\sim p}\sum_{t=1}^{R(x)}\Bigl(\TC\bigl(X_{S_t}\mid x_{C_t}\bigr)
+\sum_j\bigl(\zeta_{s_j}(x_{C_t})
-\zeta_{s_j}(x_{C_t\cup s_{<j}})\bigr)\Bigr).
\end{aligned}
\end{equation}
\end{corollary}

\begin{proof}
Both decoders visit the same batch-start states. Their likelihood factors
therefore give~\eqref{eq:seqpoint} directly from~\eqref{eq:orderedtc}.
Full support makes every factor of both likelihoods positive. Both
divergences are therefore finite, and their difference is the
$p$-expectation of~\eqref{eq:seqpoint}. This gives the first equality
in~\eqref{eq:seqdecomp}. Subtracting~\eqref{eq:learned} for the two
decoders gives the second, since singleton rounds have zero total
correlation.
\end{proof}

For a fixed block schedule whose parallel factors are the single-position
conditionals of the sequential law $P^{\mathrm{seq}}$,
\citet[Theorem~B.1]{xu2026scheduling} show that
$\KL{P^{\mathrm{seq}}}{P^{\mathrm{par}}}$ is the sum of the blocks'
conditional total correlations under $P^{\mathrm{seq}}$. This is the case
of~\eqref{eq:seqdecomp} with fixed batches, $\hat p$ the conditionals of a
joint law $r$, and $p=r$. Corollary~\ref{cor:seq} compares both
decoders to the target and allows value-adaptive batches and conditionals
that need not belong to one joint law.

\paragraph{Measurements.}
Pseudo-cost is the sequentialized-versus-parallel log-likelihood
ratio~\eqref{eq:seqpoint}, evaluated with the model's conditionals at the
ground-truth tokens; its population mean is the difference of the two
decoders' forward KL divergences (Corollary~\ref{cor:seq}). On LLaDA
WikiText, confidence top-16 exceeds random-16 by $1.503\pm0.037$
bits/token. Over the $4096$ confidence rounds ($128$ articles, $32$ rounds
each), adjacency and pseudo-cost have partial correlation $+0.47$ at a
fixed round index. The grouping measurements in Section~\ref{sec:empgen}
use LLaDA on 128
WikiText articles of 512 tokens, with position zero given and scores
normalized by the 511 predicted tokens. Confidence top-16 orders each
batch by descending confidence; random-16 uses ascending position.
Confidence-selected positions lie $0.171$ times as far from the nearest
revealed position as uniform subsets and contain $6.29$ more adjacent
pairs per round. The random schedule matches the uniform-subset comparison.

Changing the within-round confidence order to ascending position changes
the paired pseudo-cost gap from $1.5030\pm0.0371$ to
$1.5018\pm0.0374$ bits/token. Three fixed random position plans give gaps
$1.503,1.493,1.509$, positive on every article. In implementation
sensitivity checks, single precision changes fixed-schedule scores on 16
texts by at most $0.006$ bits/token. Precision and batching can change
adaptive selections and their scores by $\pm0.5$, but the paired
confidence-minus-random gap stays between $1.5$ and $2.04$.

\subsection{Task Accuracy Details}\label{app:emptask}

The four-shot prompt uses the first four GSM8K training problems and
solutions. Generated text is cut at the first special token or next
question. We extract the number after the final-answer marker, or the
last number if no marker occurs, and compare it with the reference.
Within each block, bisection starts from the revealed left neighbor and
repeatedly reveals each remaining segment's midpoint. The random rule
uses one position plan across problems and seeds. Confidence ranks
$\max_v\hat p(v\mid\cdot)$ before sampling with fixed tie-breaking.
Table~\ref{tab:taskseeds} lists the accuracy for each sampling seed and
for greedy decoding.

The first LLaDA seed seeds batches; the other two seed each problem and
rule. Primary comparisons average accuracies over seeds within each problem and
then use a paired sign-flip permutation test with 100,000 draws.
Intervals are unadjusted paired $t$ intervals; Holm correction applies to
the three pairwise $p$-values.

\begin{table}[t]
\centering
\small
\begin{tabular}{llrrrr}
\hline
model & decoding & bisection & random & confidence & reference\\
\hline
LLaDA & seed 1 & $21.5$ & $13.5$ & $6.0$ & $39.5$\\
LLaDA & seed 2 & $21.0$ & $13.5$ & $4.0$ & $41.5$\\
LLaDA & seed 3 & $23.0$ & $13.5$ & $4.0$ & $33.0$\\
LLaDA & greedy & $56.0$ & $54.5$ & $48.0$ & $70.5$\\
Dream & seed 1 & $20.0$ & $14.0$ & $2.0$ & $44.5$\\
Dream & greedy & $50.0$ & $48.0$ & $43.5$ & $71.5$\\
\hline
\end{tabular}
\caption{GSM8K accuracy in percent on the first $200$ test problems for
each sampling seed and for greedy decoding, higher is better; the rules
are those of Table~\ref{tab:task}.}
\label{tab:taskseeds}
\end{table}

Although random selection answers 27 problems correctly under every
LLaDA seed, the sets differ (42, 30, and 36 disagreements between seed
pairs). Bisection-versus-random discordant counts are $30{:}14$,
$34{:}19$, and $33{:}14$. Restricting to per-problem-seeded runs gives
$+8.5$ points ($[3.5,13.5]$, Holm-corrected $p=0.0013$); confidence remains
worse than both ($p<0.001$). Greedy differences are $1.5$
($[-6.2,9.2]$, $p=0.80$) for LLaDA and $2.0$ ($[-6.5,10.5]$, $p=0.73$)
for Dream.

On 164 HumanEval problems, the four LLaDA rules have greedy pass rates of
18.3\%, 15.9\%, 10.4\%, and 32.3\%. With temperature-1 sampling (one seed),
the rates are 3.7\%, 4.3\%, 0.6\%, and 10.4\%. Bisection and random have
$13{:}9$ and $3{:}4$ discordant outcomes, respectively.
Padding is excluded by the attention mask; Dream shifts logits by one
position and skips padding in position indices, as in its reference code.
Padded-batch versus single-prompt conditionals differ in total variation
by at most $3\times10^{-5}$ in single precision on both models, but by up
to $0.06$ in the bfloat16 arithmetic used for generation. Batch composition
can therefore change individual samples.

{\small\setlength{\bibsep}{0pt plus 0.5pt}
\bibliography{refs}}

\end{document}